\documentclass[12pt]{article}

\usepackage[utf8]{inputenc}
\usepackage[T1]{fontenc}
\usepackage{amsmath,amssymb,amsfonts,amsthm,mathtools}
\usepackage{graphicx}
\usepackage{subcaption}
\usepackage{booktabs}
\usepackage{bm}
\usepackage{bbm}
\usepackage{dsfont}
\usepackage{microtype}
\usepackage{hyperref}
\usepackage{cleveref}
\usepackage{enumitem}
\usepackage{float}
\usepackage{url}
\usepackage{nicefrac}
\usepackage{natbib}
\usepackage{appendix}
\usepackage{multirow} 
\hypersetup{colorlinks=true,citecolor=blue,linkcolor=blue}
\newtheorem{assumption}{Assumption}
\newtheorem{theorem}{Theorem}
\newtheorem{lemma}{Lemma}
\newtheorem{proposition}{Proposition}
\newtheorem{corollary}{Corollary}
\theoremstyle{remark}
\newtheorem{remark}{Remark}

\DeclareMathOperator*{\argmin}{arg\,min}
\DeclareMathOperator{\Var}{Var}
\newcommand{\E}{\mathbb{E}}
\newcommand{\Prob}{\mathbb{P}}
\newcommand{\R}{\mathbb{R}}
\newcommand{\one}{\mathbbm{1}}
\usepackage{algorithm}
\usepackage{comment}
\usepackage{xcolor}

\title{
Multi-Fidelity Quantile Regression}
\author{
Yixiang Liu
\and
Yao Zhang\thanks{Department of Statistics and Data Science, National University of Singapore.}
}

\begin{document}
\maketitle

\begin{abstract}
High-fidelity (HF) data are often expensive to collect and therefore scarce, making conditional quantiles difficult to estimate accurately. We propose a two-stage, model-agnostic method for multi-fidelity quantile regression. The central idea is a local quantile link: at each covariate value, the HF quantile is represented as a low-fidelity (LF) quantile evaluated at a covariate-dependent level. This reformulation reduces the problem to estimating the level function, which can be smoother than the HF quantile itself when the LF and HF conditional distributions have similar shapes. We also study the complementary regime in which this advantage weakens and introduce a correction step to improve robustness. Our theory characterizes when the proposed estimator converges faster than direct quantile regression using HF data alone and when the correction step provides further improvement. Experiments on synthetic and real data show that our method yields more accurate quantile estimates and tighter conformal prediction intervals.
\end{abstract}

\section{Introduction}

In many scientific applications, the most reliable measurements are also the most expensive and time-consuming to obtain. In molecular design, for example, wet-lab binding assays provide accurate estimates of binding affinity, but each experiment, from assay setup to quality control, may take days or even weeks to complete \citep{zhou2016systems,chen2023artificial}. 
Constraints of this kind are common across experimental science and form a major obstacle in areas such as drug discovery and materials design, where faster progress could have substantial societal impact by helping address urgent challenges such as major diseases and climate change \citep{vamathevan2019applications,tabor2018accelerating,butler2018machine}.
In recent years, researchers have increasingly turned to machine learning models trained on multi-fidelity data to accelerate scientific discovery \citep{kennedy2000predicting,legratiet2014recursive,wang2023exploring}. However, achieving highly accurate point prediction remains difficult because abundant low-fidelity data are only imperfect proxies for the scarce high-fidelity measurements obtained from real experiments. In settings like these, quantifying uncertainty in model predictions becomes crucial for decision-making, for example, in guiding future data collection or deciding which candidates should be validated through expensive real-world experiments \citep{yu2022uncertainty}.

In quantitative fields, quantile regression (QR) is a standard tool for uncertainty quantification \citep{koenker1978regression,koenker2005quantile}. It estimates conditional quantiles of the response given covariates and can be implemented using a range of machine learning models, such as random forests \citep{meinshausen2006quantile} and neural networks \citep{taylor2000quantile,cannon2011quantile}. QR is especially useful for modeling heteroscedastic, asymmetric, or heavy-tailed data, where lower and upper conditional quantiles can be combined to form prediction intervals for the response without specifying a parametric model for the full conditional distribution. More recently, \citet{romano2019conformalized} combined QR with conformal prediction to construct marginally valid prediction intervals. Under the standard i.i.d. assumption, these intervals are guaranteed to achieve the desired coverage on average over the population, regardless of the model used to estimate the quantiles.

However, with limited HF data, estimating conditional quantiles can be substantially harder than estimating conditional means. At a fixed covariate value \(x\), the variance of a conditional quantile estimator typically scales as
\begin{equation}\label{equ:fxy}
\mathrm{Var}\!\left(\hat Q_{Y\mid X}(\tau\mid x)\right)
\approx
\frac{\tau(1-\tau)}{n\, f_{Y\mid X}\!\left(Q_{Y\mid X}(\tau\mid x)\mid x\right)^2},
\end{equation}
where \(n\) is the sample size, and \(f_{Y\mid X}(\cdot\mid x)\) is the conditional density evaluated at the target quantile \citep{koenker2005quantile}. When both the sample size and this density are small, the fitted quantiles can be highly variable. Even when conformal prediction guarantees marginal coverage, the resulting intervals may still be wide or unstable if the fitted quantiles are noisy. This makes it important to exploit auxiliary information whenever it is available.

Although HF data are scarce, low-fidelity (LF) data can often be obtained at scale. In molecular design, for example, molecular docking or lower-level quantum calculations provide cheaper but less accurate approximations to binding affinity \citep{pantsar2018binding,kairys2019binding}. Multi-fidelity learning seeks to leverage such abundant LF observations to compensate for the scarcity of HF data \citep{peherstorfer2018survey}. However, most existing multi-fidelity methods focus on conditional means or other average response functionals \citep{kennedy2000predicting,legratiet2014recursive,perdikaris2017nonlinear,meng2020composite,fernandez2023review}, and extending these ideas to quantile regression is less straightforward. Unlike conditional means, conditional quantiles can differ across fidelity levels in nonlinear, level-dependent ways. For instance, even if molecular docking scores track average binding affinity well after a simple rescaling, their conditional distribution may be substantially noisier or more skewed than the wet-lab measurements. Consequently, it is not obvious which LF quantile level should correspond to a given HF quantile, or how to correct this mismatch without imposing strong parametric assumptions.

\subsection{Overview}
Motivated by these challenges, we propose Multi-Fidelity Quantile Regression (MFQR), which rests on the following link between the two fidelity levels:

\begin{assumption}[Local Quantile Link]\label{assume:qq}
For a fixed target level \(\tau\in (0,1)\) and any covariate value \(x \in \mathcal X\), the HF and LF quantile functions satisfy
\begin{equation}\label{equ:qq}
Q^H(\tau\mid x)=Q^L(u_\tau(x)\mid x),
\end{equation}
for some covariate-dependent level function \(u_\tau(x)\in(0,1)\).
\end{assumption}
Assumption~\ref{assume:qq} means that the target HF quantile can be matched locally by an LF quantile at a different, covariate-dependent probability level \(u_\tau(x)\). 
For instance, if the LF distribution is heavier-tailed than the HF distribution, then \(u_\tau(x)\) should lie closer to \(0.5\) than \(\tau\), so that a more central LF quantile matches the more extreme HF quantile. 
Thus, the assumption is mild whenever the support of the LF conditional distribution contains the target HF quantile. 
In particular, it holds formally for common full-support location-scale models, such as Gaussian, Student-\(t\), logistic, or Laplace conditional distributions. 
If the two conditional distributions are centered at different locations, we can first center the responses by subtracting conditional mean estimators and then apply the same quantile-link formulation to the resulting residuals.

When the two fidelity levels have similar conditional shapes, the shifted probability level $u_\tau(x)$ may vary more smoothly over the covariate space than the original HF quantile function. This suggests that reformulating quantile estimation on the probability scale can reduce the sample complexity of estimating the high-fidelity quantile. Motivated by this idea, we propose a model-agnostic wrapper that uses the low-fidelity distribution to learn the level function $u_\tau(x)$, together with a correction step based on the high-fidelity quantile equation to refine the quantile estimate.
Our theory studies when the wrapper estimator converges faster than direct HF quantile regression and when the correction step can yield further improvements. Experiments on synthetic and real datasets show that the proposed method produces more accurate quantile estimates and narrower conformal intervals than baseline methods while preserving marginal coverage guarantees.

The remainder of this paper is organized as follows. 
\Cref{sect:related} reviews related work. 
\Cref{sec:wrapper_method} introduces the wrapper estimator and establishes its theoretical convergence properties. 
\Cref{sec:onestep} develops a correction framework for refining the wrapper estimator.
Finally, \Cref{sec:experiments} demonstrates the empirical performance of MFQR on synthetic and scientific datasets.

\section{Related Work}\label{sect:related}

This paper relates to several strands of literature: multi-fidelity modeling, transfer learning, quantile regression, and conformal prediction.

A classical starting point for multi-fidelity modeling is the autoregressive framework of \citet{kennedy2000predicting}, which models the high-fidelity (HF) response as a scaled version of the low-fidelity (LF) response plus a discrepancy term. This framework has been extended to multiple fidelity levels \citep{legratiet2014recursive} and to nonlinear cross-fidelity mappings using Gaussian processes and neural networks \citep{perdikaris2017nonlinear,meng2020composite,fernandez2023review}. Most of this literature, however, focuses on conditional means, surrogate response surfaces, or other average response functionals. By contrast, we study nonparametric transfer of conditional quantiles across fidelity levels, without assuming simple autoregressive relationships or Gaussian response models. Our approach is complementary to existing multi-fidelity methods: one may first estimate and remove shared mean structure across fidelity levels, and then apply our quantile-transfer procedure to the residuals.

Our work is also related to transfer learning for quantile regression and transformation-based regression. In parametric settings, \citet{huang2022estimation}, \citet{zhang2022transfer}, and \citet{jin2024transfer} study transfer learning for high-dimensional linear quantile regression models by exploiting approximate similarity between source and target coefficient vectors. Classical transformation-based methods use parametric or nonparametric transformations to simplify regression structure \citep{box1964analysis,breiman1985estimating}, and \citet{fan2016multitask} developed multitask quantile regression under a transnormal model using rank-based covariance regularization. We do not assume that a global marginal transformation makes the joint distribution Gaussian, nor do we transfer linear coefficients across tasks. Instead, we transfer quantile information through a covariate-dependent level function that specifies which LF quantile locally corresponds to a target HF quantile. This formulation directly addresses a cross-fidelity matching problem that does not arise in the settings considered by coefficient-transfer or global transformation approaches.

Our quantile-link idea is also conceptually related to the changes-in-changes literature in causal inference \citep{athey2006identification,callaway2019quantile,melly2015changes}, which uses distributional transformations to identify counterfactual outcome distributions in nonlinear difference-in-differences models. The goal of these methods, however, is to identify counterfactual outcome distributions and distributional treatment effects. By contrast, the key estimand in our framework is a covariate-dependent level function. Our analysis studies when this transformation is useful and when it may fail. In particular, the LF transformation alone can be unhelpful or even misleading when the LF and HF distributions differ in ways that make the induced level function hard to estimate. This motivates our correction step, which refines the wrapper estimator using the HF quantile equation.

The correction step in our method draws on the literature on debiased quantile regression. 
In linear quantile regression, \citet{wang2024semiparametric} propose a computationally feasible one-step estimator based on semiparametric efficient scores. 
In distributed settings, \citet{pan2022note} develop a communication-efficient one-step update for quantile regression. 
In high-dimensional models, \citet{giessing2023debiased} study debiasing methods for penalized quantile regression based on rank scores and smoothing. 
These works use first-order corrections to improve the efficiency or scalability of pilot quantile estimators. 
Our goal is different: we use a one-step correction to refine the proposed estimator in multi-fidelity settings, particularly when the wrapper loses its advantage or \Cref{assume:qq} holds only approximately.

Finally, our work is connected to conformal prediction through conformalized quantile regression (CQR) \citep{romano2019conformalized}, which calibrates estimated quantiles on a hold-out set to obtain marginally valid prediction intervals. More broadly, our method relates to conformal prediction methods that adapt prediction-set or interval levels to calibrate classifiers \citep{zhang2024posterior}, restore marginal coverage after localization \citep{guan2023localized}, or handle distribution shift \citep{gibbs2021adaptive}.
These methods pursue goals different from ours: none studies multi-fidelity learning, and none seeks to learn a covariate-dependent cross-fidelity level function. We are also aware of recent work \citep{moya2024conformal} applying split conformal prediction \citep{papadopoulos2002inductive,lei2018distribution} to multi-fidelity DeepONet regression \citep{howard2022multifidelity,lu2022multifidelity}, which is aimed at physics applications involving mappings between infinite-dimensional function spaces. By contrast, we study multi-fidelity quantile regression directly, with the broader goal of improving uncertainty quantification in data-scarce settings.

\section{Multi-fidelity quantile regression (MFQR)}
\label{sec:wrapper_method}

We focus on the two-fidelity setting. Let
\[
\mathcal D^H=\{(X_i^H,Y_i^H)\}_{i=1}^{n_H},
\qquad
\mathcal D^L=\{(X_j^L,Y_j^L)\}_{j=1}^{n_L}
\]
be independent high-fidelity (HF) and low-fidelity (LF) samples, with \(n_L \gg n_H\). 
Let \(F^{H}(y\mid x)\) and \(F^L(y\mid x)\) denote the conditional HF and LF cumulative distribution functions (CDFs), respectively.
Our goal is to estimate the HF conditional quantile $q_\tau(x) := Q^H(\tau\mid x)$ for a fixed level $\tau\in(0,1)$.

In practice, it is often convenient to assume that \(Y^H\) and \(Y^L\) have the same conditional mean \(\mu(x)\), or more generally, as in \citet{kennedy2000predicting}, that their mean functions \(\mu_L(x)\) and \(\mu_H(x)\) satisfy a simple autoregressive relationship. Subtracting mean estimators from the responses adjusts for mean differences. Then MFQR can be applied to the resulting residuals.

In this section, we impose the following standing assumption.
Let \(\mathcal C^m(A)\) denote the class of functions on a set \(A\subset\mathbb R\) with continuous derivatives up to order \(m\). We write \(\lceil \beta_\rho\rceil\) for the smallest integer not smaller than \(\beta_\rho\).

\begin{assumption}
\label{ass:global_locality}
Assume that \(\mathcal X\subset\mathbb R^p\) is compact. There exist compact intervals
\(
\mathcal Y\subset\mathbb R
\)
and
\(
\mathcal U^L\subset(0,1)
\)
such that, for all \(x\in\mathcal X\),

\begin{enumerate}[label=(\roman*)]
    \item the target HF quantile \(q_\tau(x)\) lies in the interior of \(\mathcal Y\), the wrapped target level \(u_\tau(x)\) lies in the interior of \(\mathcal U^L\), and the conditional support of \(U:=F^L(Y^H\mid X)\) given \(X=x\) is contained in \(\mathcal U^L\);

    \item the conditional supports of both \(Y^H\mid X=x\) and \(Y^L\mid X=x\) are contained in \(\mathcal Y\);

    \item \(F^H(\cdot\mid x)\in \mathcal C^1(\mathcal Y)\) and \(F^L(\cdot\mid x)\in \mathcal C^1(\mathcal Y)\), with
    \[
    f^H(y\mid x)\ge c_H>0,~\forall y\in\mathcal Y,
    \qquad
    f^L(y\mid x)\ge c_L>0,~\forall y\in\mathcal Y;
    \]

    \item \(Q^L(\cdot\mid x)\in \mathcal C^2(\mathcal U^L)\).
\end{enumerate}
\end{assumption}

This standing assumption is imposed to simplify the inversion arguments used
throughout the theory. It guarantees that the relevant HF and LF quantiles lie in
regions where the conditional CDFs are strictly increasing with densities bounded
away from zero, so that the associated quantile maps are locally stable and
differentiable. In particular, it avoids repeatedly introducing
\(x\)-dependent neighborhoods around the target quantiles in later results. For
full-support distributions, the same arguments can be localized to compact
neighborhoods of the target quantiles on which the relevant densities are
bounded away from zero.

\subsection{The wrapper estimator}

Recall that our key idea in the introduction was that, under Assumption~\ref{assume:qq}, the target HF quantile can be represented as an LF quantile evaluated at a covariate-dependent probability level \(u_\tau(x)\). The next result makes this representation operational by showing that \(u_\tau(x)\) is itself a conditional quantile of a transformed response. This allows us to estimate \(u_\tau(x)\) directly by quantile regression on the transformed probability scale.

\begin{proposition}\label{prop:u}
Under Assumptions \ref{assume:qq} and \ref{ass:global_locality}, the level function \(u_\tau(x)\) in \eqref{equ:qq} is the conditional \(\tau\)-quantile of
\(
U:=F^L(Y^H\mid X),
\)
that is, for any $x\in \mathcal X$,
\[
u_\tau(x) = Q_{U\mid X}(\tau\mid x).
\]
\end{proposition}

\begin{proof}
Fix \(x\in\mathcal X\). By Assumption~\ref{assume:qq},
\(
q_\tau(x)=Q^L(u_\tau(x)\mid x).
\)
Since \(F^L(\cdot\mid x)\) is strictly increasing on \(\mathcal Y\) by
Assumption~\ref{ass:global_locality}, applying \(F^L(\cdot\mid x)\) to both
sides gives
\(
u_\tau(x)=F^L(q_\tau(x)\mid x).
\)
Now define \(U=F^L(Y^H\mid X)\). Conditional on \(X=x\), this becomes
\(U=F^L(Y^H\mid x)\). Therefore,
\begin{align*}
\mathbb P\{U\le u_\tau(x)\mid X=x\}
&=
\mathbb P\{F^L(Y^H\mid x)\le u_\tau(x)\mid X=x\} \\
&=
\mathbb P\{F^L(Y^H\mid x)\le F^L(q_\tau(x)\mid x)\mid X=x\} \\
&=
\mathbb P\{Y^H\le q_\tau(x)\mid X=x\} \\
&=
F^H(q_\tau(x)\mid x) \\
&=
F^H(Q^H(\tau\mid x)\mid x) \\
&=
\tau .
\end{align*}
Here the first equality uses the definition of \(U\); the second uses
\(u_\tau(x)=F^L(q_\tau(x)\mid x)\); the third uses the strict monotonicity of
\(F^L(\cdot\mid x)\); the fourth is the definition of the conditional CDF
\(F^H\); the fifth uses \(q_\tau(x)=Q^H(\tau\mid x)\); and the last equality
uses the fact that \(F^H(\cdot\mid x)\) is strictly increasing in a neighborhood
of \(q_\tau(x)\), so that \(F^H(Q^H(\tau\mid x)\mid x)=\tau\).

It remains to justify uniqueness of the conditional \(\tau\)-quantile of
\(U\mid X=x\). For \(u\in\mathcal U^L\), define
\[
H_U(u\mid x):=\mathbb P\{U\le u\mid X=x\}.
\]
Using again the definition of \(U\) and the strict monotonicity of
\(F^L(\cdot\mid x)\),
\[
H_U(u\mid x)
=
\mathbb P\{Y^H\le Q^L(u\mid x)\mid X=x\}
=
F^H(Q^L(u\mid x)\mid x).
\]
By Assumption~\ref{ass:global_locality}, \(F^H(\cdot\mid x)\) and
\(F^L(\cdot\mid x)\) are continuously differentiable on \(\mathcal Y\), and
their densities are bounded away from zero there. Hence, on the relevant range,
\(H_U(\cdot\mid x)\) is differentiable with density
\[
f_{U\mid X}(u\mid x)
=
\frac{
f^H(Q^L(u\mid x)\mid x)
}{
f^L(Q^L(u\mid x)\mid x)
}.
\]
This density is positive for all \(u\in\mathcal U^L\). Thus
\(H_U(\cdot\mid x)\) is strictly increasing on the relevant neighborhood, so the
conditional \(\tau\)-quantile of \(U\mid X=x\) is unique. Since we have already
shown that
\[
\mathbb P\{U\le u_\tau(x)\mid X=x\}=\tau,
\]
it follows that
\[
Q_{U\mid X}(\tau\mid x)=u_\tau(x).
\]
\end{proof}

\begin{remark}
\Cref{prop:u} also follows from a stronger but more interpretable condition on the CDFs. Suppose there exists an increasing map
\(
g_x:[0,1]\to[0,1]
\)
with \(g_x(0)=0\) and \(g_x(1)=1\) such that
\[
F^{H}(y\mid x)=g_x\!\left(F^L(y\mid x)\right),\qquad \forall y\in\R.
\]
Then the LF and HF distributions differ only through a monotone distortion on the probability scale. If \(F^H(\cdot\mid x)\) is strictly increasing, then \Cref{assume:qq} holds with
\(
u_{\tau}(x)=g_x^{-1}(\tau).
\)
This condition is stronger because it specifies a full distributional relationship between the conditional CDFs \(F^{L}\) and \(F^{H}\).
\end{remark}

\begin{algorithm}[t]
\caption{Wrapper estimator in MFQR}
\label{alg:wrapper}
\vspace{-10pt}
\begin{enumerate}[leftmargin=*,label=\arabic*.]
    \item Fit an estimator \(\hat F^L(\cdot\mid x)\) of the LF CDF \(F^L(\cdot\mid x)\) using the sample \(\mathcal D^L\).
    \item For each HF observation \((X_i^H,Y_i^H)\), compute the wrapped pseudo-response
    \[
    \hat U_i := \hat F^L(Y_i^H\mid X_i^H),\qquad i=1,\dots,n_H.
    \]
    \item Use the HF sample \(\{(X_i^H,\hat U_i)\}_{i=1}^{n_H}\) to estimate the conditional quantile function \(u_\tau(x)=Q_{U\mid X}(\tau\mid x)\), yielding an estimator \(\hat u_\tau(x)\).
    \item Map back through the LF function to obtain the wrapper estimator
    \[
    \tilde q_\tau(x):=(\hat F^L)^{-1}(\hat u_\tau(x)\mid x).
    \]
\end{enumerate}
\vspace{-10pt}
\end{algorithm}

\Cref{alg:wrapper} summarizes the construction of the wrapper estimator \(\tilde q_\tau(x)\) for the HF quantile \(Q^H(\tau\mid x)\). To make the procedure concrete, one simple implementation is kernel-based. First, we estimate the LF conditional CDF \(F^L(y\mid x)\) by a Nadaraya--Watson estimator,
\begin{equation}
\label{equ:F_L_hat}
\hat F^L(y\mid x)
=
\sum_{j=1}^{n_L} w_j^L(x)\,\one\{Y_j^L\le y\},
\qquad
w_j^L(x)
:=
\frac{K_{h_L}(X_j^L-x)}
{\sum_{i=1}^{n_L} K_{h_L}(X_i^L-x)},
\end{equation}
where \(K_{h_L}(z):=h_L^{-p}K(z/h_L)\) is a kernel with bandwidth \(h_L>0\).

We then wrap the HF responses by computing
\(
\hat U_i:=\hat F^L(Y_i^H\mid X_i^H).
\)
Next, we estimate \(u_\tau(x)\) by the local constant kernel quantile estimator
\begin{equation}
\label{equ:local_quantile}
\hat u_\tau(x)
=
\inf\left\{
u:
\hat F^U(u\mid x):=
\sum_{i=1}^{n_H} w_i^U(x)\,\one\{\hat U_i\le u\}\ge \tau
\right\},
\end{equation}
where the weights $w_i^U(x)$ are defined as
\[
w_i^U(x)
:=
\frac{K_{h_u}(X_i^H-x)}
{\sum_{\ell=1}^{n_H} K_{h_u}(X_\ell^H-x)},
\quad
K_{h_u}(z):=h_u^{-p}K(z/h_u).
\]
Finally, we obtain the wrapper estimator \(\tilde q_\tau(x)\) by mapping \(\hat u_\tau(x)\) back to the response scale through the inverse LF quantile function:
\[
\tilde q_\tau(x):=(\hat F^L)^{-1}(\hat u_\tau(x)\mid x).
\]
\textbf{Other models.}
The kernel estimators above provide one concrete implementation of the wrapper estimator. In Appendix \ref{app:implementation}, we discuss a few machine learning approaches for estimating CDFs and quantiles.

\subsection{Why the wrapped target can be smoother}
\label{sec:wrapper_smoothness}

To understand why the wrapped target $u_\tau(x)$ is often easier to estimate than the original HF quantile $q_\tau(x)$, consider a structural model where the two fidelities differ only by a localized scale distortion $\rho(x)$:
\begin{equation}\label{eq:two_model}
Y^H = \mu(x) + \sigma(x)W^H, \qquad Y^L = \mu(x) + \sigma(x)\rho(x)W^L,
\end{equation}
where $W^H$ and $W^L$ are independent noise variables with strictly increasing, and potentially distinct, CDFs $G_H$ and $G_L$, respectively. Here \(\mu(x)\) is a common mean function, \(\sigma(x)\) is a common baseline scale, and \(\rho(x)\) describes the relative scale distortion between the two fidelities.

Under model \eqref{eq:two_model}, the target HF quantile function can be written as
\[
q_\tau(x)=\mu(x)+\sigma(x)G_H^{-1}(\tau).
\]
The mean-adjusted target is defined as
\[
r_\tau(x):=q_\tau(x)-\mu(x)=\sigma(x)G_H^{-1}(\tau).
\]
The wrapper target is defined as
\[
u_\tau(x):=F^L(q_\tau(x)\mid x)
=
G_L\!\left(\frac{r_\tau(x)}{\sigma(x)\rho(x)}\right)
=
G_L\!\left(\frac{G_H^{-1}(\tau)}{\rho(x)}\right).
\]
Thus the wrapped target \(u_\tau(x)\) depends only on the relative distortion \(\rho(x)\), but not on the mean function \(\mu(x)\) or the baseline scale \(\sigma(x)\).

Suppose the mean, scale, and distortion functions are Hölder smooth, that is, they belong to Hölder classes \(\mathcal H(\beta,C)\). Roughly speaking, \(f\in \mathcal H(\beta,C)\) means that \(f\) has smoothness order \(\beta\) and Hölder constant \(C\), so larger \(\beta\) corresponds to a smoother function. A definition of \(\mathcal H(\beta,C)\) is given in Appendix~\ref{section:holder}.

\begin{assumption}
\label{ass:holder_mu}
The mean, scale, and distortion functions in \eqref{eq:two_model} satisfy
\[
\mu \in \mathcal H(\beta_\mu,C_\mu),\qquad
\sigma \in \mathcal H(\beta_\sigma,C_\sigma),\qquad
\rho \in \mathcal H(\beta_\rho,C_\rho).
\]
\end{assumption}
\vspace{5pt}

\begin{assumption}
\label{ass:G_regular}
For fixed \(\tau\in(0,1)\), the quantile \(G_H^{-1}(\tau)\) is finite, \(\rho(x)\) is bounded away from zero on \(\mathcal X\), and
\(
G_L \in \mathcal C^{\lceil \beta_\rho\rceil}(\mathbb R).
\)
\end{assumption}
\vspace{5pt}

\begin{proposition}
\label{prop:wrapper_smoother_extended}
Under Assumptions~\ref{ass:holder_mu}--\ref{ass:G_regular},
\[
q_\tau \in \mathcal H\!\left(\min\{\beta_\mu,\beta_\sigma\},\; C_q\right),
\quad
r_\tau \in \mathcal H\!\left(\beta_\sigma,\; |G_H^{-1}(\tau)|C_\sigma\right),
\quad
u_\tau \in \mathcal H\!\left(\beta_\rho,\; C_\tau C_\rho\right),
\]
where \(C_q>0\) is a finite constant depending only on \(\tau,G_H,C_\mu,C_\sigma\), and \(C_\tau>0\) is a finite constant depending only on \(\tau,G_H,G_L\), and the range of \(\rho\).
\end{proposition}

Proposition~\ref{prop:wrapper_smoother_extended} explains the basic rationale for using the wrapper. The raw HF quantile \(q_\tau(x)\) inherits variation from both the mean function \(\mu(x)\) and the scale function \(\sigma(x)\), while the mean-adjusted target \(r_\tau(x)\) still reflects the full local scale variation of the HF distribution. By contrast, the wrapped target \(u_\tau(x)\) depends only on the relative distortion \(\rho(x)\) between the two fidelities.

When the distortion function \(\rho(x)\) is smoother than the components driving \(q_\tau(x)\), the wrapper turns the original problem into a lower-complexity regression task. Equivalently, if
\(
\beta_\rho>\min\{\beta_\mu,\beta_\sigma\},
\)
then the wrapped target $u_\tau$ is smoother, and therefore easier to estimate, than the original HF quantile $q_{\tau}$.

\subsection{How this gain transfers to the wrapper estimator}

We next show that any estimation gain for the wrapped target \(u_\tau(x)\) carries over directly to the wrapper estimator \(\tilde q_\tau(x)\), up to the LF plug-in error. This transfer relies on a mild local regularity condition on the LF distribution.

\begin{theorem}
\label{thm:wrapper_nonparam}
Fix \(x\in\mathcal X\) and \(\tau\in(0,1)\). Under Assumptions~\ref{assume:qq} and \ref{ass:global_locality}, suppose \(\hat F^L(\cdot\mid x)\) is nondecreasing in \(y\) with probability tending to one. Let \(a_n,b_n\downarrow 0\) be deterministic sequences such that
\[
\hat u_\tau(x)-u_\tau(x)=O_p(a_n),
\qquad
\sup_{y\in \mathcal Y}\bigl|\hat F^L(y\mid x)-F^L(y\mid x)\bigr|=O_p(b_n).
\]
Then
\[
\tilde q_\tau(x)-q_\tau(x)=O_p(a_n+b_n).
\]
\end{theorem}

To make Theorem~\ref{thm:wrapper_nonparam} concrete, we specialize to the wrapper estimator \(\tilde q_\tau\) defined through the kernel-based CDF estimator \(\hat F^L\) in \eqref{equ:F_L_hat} and the kernel-based wrapped-target estimator \(\hat u_\tau\) in \eqref{equ:local_quantile}. We let  $F^U(u\mid x)$
denote the conditional CDF of the wrapped response \(U=F^L(Y^H\mid X)\).

We compare the wrapper estimator \(\tilde q_\tau(x)\) with the direct estimator \(\hat q_\tau^H(x)\), defined in the same way as in \eqref{equ:local_quantile} but with \(\hat U_i\) replaced by \(Y_i^H\). The following are standard assumptions in kernel smoothing theory \citep{stone1982optimal,chaudhuri1991nonparametric,fan1992design,fan1994robust,yu1998local}. 
Instead of imposing H\"older smoothness on the covariate density, we assume boundedness of the design density together with a kernel-mass condition ensuring that local neighborhoods have enough mass, including near the boundary of \(\mathcal X\).

\begin{assumption}
\label{ass:kernel_rate_setup}
In the setup of Theorem~\ref{thm:wrapper_nonparam}, for 
\(\beta_q,\beta_u,\beta_L\in(0,1]\),
\begin{enumerate}[label=(\roman*)]
    \item the direct HF target and its conditional CDF satisfy
    \[
    q_\tau \in \mathcal H(\beta_q,C_q),
    \qquad
    F^H(y\mid \cdot)\in \mathcal H(\beta_q,C_q^F)\quad \text{for all } y\in \mathcal Y;
    \]
    \item the wrapped target and the wrapped response distribution satisfy
    \[
    u_\tau \in \mathcal H(\beta_u,C_u),
    \qquad
    F^U(u\mid \cdot)\in \mathcal H(\beta_u,C_u^F)\quad \text{for all } u\in \mathcal U^L,
    \]
    and \(F^U(\cdot\mid x)\in \mathcal C^1(\mathcal U^L)\) with density \(f^U(u\mid x)\ge c_U>0\) for all \(u\in \mathcal U^L\);
    \item the LF conditional CDF satisfies
    \[
    F^L(y\mid \cdot)\in \mathcal H(\beta_L,C_L)\quad \text{for all } y\in \mathcal Y;
    \]
    \item the covariate distribution has a density \(f_X\) satisfying
    \[
    0<c_X\le f_X(x)\le C_X<\infty
    \quad \text{for all } x\in\mathcal X,
    \]
    and, for the kernel \(K\) used in part (v), there exist constants \(h_0>0\) and \(c_{\mathcal X}>0\) such that
    \[
    \inf_{0<h<h_0}\inf_{x\in\mathcal X}
    \int_{\mathcal X} h^{-p}K((t-x)/h)\,dt
    \ge c_{\mathcal X};
    \]
    \item the estimators \(\hat q_\tau^H\), \(\hat u_\tau\), and \(\hat F^L\) described above use a bounded, Lipschitz, compactly supported, nonnegative kernel \(K:\mathbb R^p\to\mathbb R\) satisfying
    \[
    \int K(u)\,du=1,
    \qquad
    \int K(u)^2\,du<\infty,
    \]
    with bandwidths \(h_q,h_u,h_L\) such that
    \[
    h_q,h_u,h_L\to 0,
    \qquad
    n_H h_q^p\to\infty,\quad
    n_H h_u^p\to\infty,\quad
    n_L h_L^p/\log n_L \to\infty.
    \]
\end{enumerate}
\end{assumption}
\vspace{5pt}

\begin{corollary}
\label{cor:wrapper_rate_transfer}
Under Assumptions~\ref{assume:qq}, \ref{ass:global_locality}, and \ref{ass:kernel_rate_setup}, using the bandwidths
\[
h_q^\star \asymp n_H^{-1/(2\beta_q+p)},\qquad
h_u^\star \asymp n_H^{-1/(2\beta_u+p)},\qquad
h_L^\star \asymp \left(\frac{\log n_L}{n_L}\right)^{1/(2\beta_L+p)},
\]
the kernel-based direct HF estimator satisfies
\[
\hat q_\tau^H(x)-q_\tau(x)
=
O_p\!\left(n_H^{-\beta_q/(2\beta_q+p)}\right),
\]
while the kernel-based wrapper estimator satisfies
\[
\tilde q_\tau(x)-q_\tau(x)
=
O_p\!\left(
n_H^{-\beta_u/(2\beta_u+p)}
+
\left(\frac{\log n_L}{n_L}\right)^{\beta_L/(2\beta_L+p)}
\right).
\]
\end{corollary}

The bandwidth choice \(h_L^\star\) for \(\hat F^L\) differs from those for \(\hat q_\tau^H\) and \(\hat u_\tau\) because \(\hat F^L\) is used to generate the pseudo-responses \(\hat U_i=\hat F^L(Y_i^H\mid X_i^H)\) for all HF observations. To control this first-stage error, we require a uniform convergence rate for \(\hat F^L(y\mid x)\) over both \(x\) and \(y\), which introduces the additional \(\log n_L\) factor in \(h_L^\star\). By contrast, the estimators \(\hat q_\tau^H\) and \(\hat u_\tau\) are analyzed pointwise at a fixed covariate value \(x\), so their bandwidths follow the usual local bias--variance tradeoff without this logarithmic term.

If the LF plug-in error is negligible relative to the error of \(\hat u_\tau(x)\), for example when \(n_L\) is sufficiently larger than \(n_H\), then
\[
\tilde q_\tau(x)-q_\tau(x)
=
O_p\!\left(n_H^{-\beta_u/(2\beta_u+p)}\right).
\]
If \(\beta_u>\beta_q\), then the optimal bandwidth \(h_u^\star\) is asymptotically larger than \(h_q^\star\). Under the structural model \eqref{eq:two_model}, Proposition~\ref{prop:wrapper_smoother_extended} gives \(\beta_u=\beta_\rho\), so this condition becomes \(\beta_\rho>\beta_q=\min\{\beta_\mu,\beta_\sigma\}.\)
The larger bandwidth $h_u^\star$ yields more diffuse kernel weights in \eqref{equ:local_quantile}, with effective local sample size
\[
n_{\mathrm{eff}}(x;h_u^\star):=
\left[\sum_{i=1}^{n_H} w_i^U(x;h_u^\star)^2\right]^{-1}
\asymp n_H (h_u^\star)^p.
\]
Hence the wrapper estimator uses \(n_{\mathrm{eff}}(x;h_u^\star)\) effective observations in each local fit, compared with \(n_{\mathrm{eff}}(x;h_q^\star)\) for the direct estimator, and therefore has a smaller leading variance term. In this sense, the wrapper mitigates the high-variance difficulty of direct quantile regression highlighted around \eqref{equ:fxy}.

\textbf{Other models.} Our formal rate analysis is derived for the local constant kernel implementation. Nevertheless, the same high-level intuition extends to other estimators that can be viewed through adaptive local weights. For example, quantile regression forests \citep{meinshausen2006quantile} produce predictions through weighted empirical distributions, where the weights are induced by the forest partition. When the wrapped target is smoother than the original HF quantile, one expects that accurate prediction can be achieved with less aggressive partitioning of the feature space, for example through larger terminal nodes or shallower trees. This leads to more diffuse weights and hence lower variance, which is the forest analogue of using a larger bandwidth in kernel smoothing. We emphasize, however, that our formal theory is stated only for the kernel estimator; extending these guarantees to adaptive machine learning estimators is an important direction for future work.

\section{Wrapper correction via a proximal framework}
\label{sec:onestep}

We now consider the complementary regime in which the wrapper loses its advantage. This can happen when the LF transformation does not simplify the estimation problem enough, so that \(u_\tau(x)\) is no easier to estimate than \(q_\tau(x)\), or when the LF plug-in error is not negligible. In either case, relying on the wrapper alone may be suboptimal.
Furthermore, we show how to correct the wrapper estimator even when Assumption \ref{assume:qq} does not hold.

\subsection{Correction strategies}\label{sect:corrections}

We structure our correction strategies using a proximal optimization framework \citep{rockafellar1976monotone,parikh2014proximal}. In its original form, a proximal method stabilizes optimization by penalizing large departures from a current iterate. We adopt this principle to design corrections: we update toward the high-fidelity quantile equation while anchoring to the stable wrapper estimator. Specifically, we consider corrected estimators of the form
\begin{equation}
\label{eq:prox_general}
\tilde q_\tau^{+}(x)
\in
\argmin_{q\in\R}
\Bigl\{
\mathcal S_x\bigl(q,\tilde q_\tau(x)\bigr)
+
\lambda\,\mathcal E_x(q)
\Bigr\},
\end{equation}
where \(\mathcal S_x\) is a stability term and \(\mathcal E_x\) is a fidelity term based on the HF information. The first term keeps the updated estimator close to the stable, low-variance wrapper estimator \(\tilde q_\tau(x)\), while the second term encourages better agreement with the HF quantile equation. 

This proximal framework includes several natural correction strategies.

\textbf{Mixed estimator.} We first consider setting
\[
\mathcal S_x\bigl(q,\tilde q_\tau(x)\bigr)=(q-\tilde q_\tau(x))^2,
\qquad
\mathcal E_x(q)=(q-\hat q_\tau^H(x))^2.
\]
Then the minimizer of \eqref{eq:prox_general} takes the form
\[
\tilde q_\tau^{\,\mathrm{mix}}(x)
=
\omega\hat q_\tau^H(x) + (1-\omega)\tilde q_\tau(x),
\]
where \(\omega:=\lambda/(1+\lambda)\in[0,1)\). Thus, \(\tilde q_\tau^{\,\mathrm{mix}}(x)\) is exactly a convex combination of the wrapper estimator and the direct HF estimator.

\textbf{Projection estimator.} Using the HF quantile equation in \(\mathcal E_x(q)\) leads to
\begin{equation}
\label{eq:projection_estimator}
\hat q_{\tau,\lambda}^{\,\mathrm{proj}}(x)
\in
\argmin_{q\in\R}
\Bigl\{
(q-\tilde q_\tau(x))^2
+
\lambda\bigl(\hat F^{H}(q\mid x)-\tau\bigr)^2
\Bigr\}.
\end{equation}
This may be viewed as a soft projection of the wrapper estimator toward the set of values satisfying the HF quantile equation. Unlike convex blending, \(\hat q_{\tau,\lambda}^{\,\mathrm{proj}}(x)\) does not interpolate directly between two estimators. Instead, it balances two objectives: staying close to the stable wrapper estimator and improving agreement with the HF quantile equation.

\textbf{One-step estimator.}
To obtain a simple closed-form correction, we linearize the HF quantile equation around the wrapper estimator \(q=\tilde q_\tau(x)\):
\[
\hat F^{H}(q\mid x)-\tau
\approx
\hat F^{H}(\tilde q_\tau(x)\mid x)-\tau
+
\hat f^{H}(\tilde q_\tau(x)\mid x)\bigl(q-\tilde q_\tau(x)\bigr).
\]
Substituting this approximation into \eqref{eq:projection_estimator} and minimizing yields
\[
\tilde q_\tau^{+}(x)
=
\tilde q_\tau(x)
-
\gamma_\lambda(x)\,
\frac{\hat F^{H}(\tilde q_\tau(x)\mid x)-\tau}
{\hat f^{H}(\tilde q_\tau(x)\mid x)},
\quad
\gamma_\lambda(x)
=
\frac{\lambda\,\hat f^{H}(\tilde q_\tau(x)\mid x)^2}
{1+\lambda\,\hat f^{H}(\tilde q_\tau(x)\mid x)^2}.
\]
Thus, the one-step correction may be viewed as the local linearization of the projection rule \eqref{eq:projection_estimator}. Since \(\gamma_\lambda(x)\in[0,1)\), it is natural to work directly with a step-size parameter \(\gamma\in[0,1]\) and define the one-step estimator as
\begin{equation}\label{equ:q_tau_gamma}
\tilde q_{\tau,\gamma}(x)
=
\tilde q_\tau(x)
-
\gamma\,
\frac{\hat F^{H}(\tilde q_\tau(x)\mid x)-\tau}
{\hat f^{H}(\tilde q_\tau(x)\mid x)}.
\end{equation}
Finally, we turn to an asymptotic analysis of the one-step estimator.

\begin{assumption}
\label{ass:H_quantile_reg}
For each fixed \(x\), the map \(y\mapsto F^{H}(y\mid x)\) is twice continuously differentiable in \(y\), with conditional density
\[
f^{H}(y\mid x):=\partial_y F^{H}(y\mid x),
\qquad
f^{H}(q_\tau(x)\mid x)\ge c_H>0.
\]
\end{assumption}
Assumption~\ref{ass:H_quantile_reg} is a local regularity condition for the HF conditional distribution. 
It requires the HF CDF to be twice differentiable near the target quantile and the conditional density at the target quantile to be bounded away from zero. 
This ensures that the HF quantile is locally identifiable and that Taylor expansions of the HF quantile equation are stable.

\begin{assumption}
\label{ass:pilot_local}
For each fixed \(x\in\mathcal X\) and \(\tau\in(0,1)\),
\begin{align*}
& e_\tau(x) := \tilde q_\tau(x) - q_{\tau}(x)=o_p(1),\\[5pt]
& \eta_\tau(x):=\hat f^{H}(\tilde q_\tau(x)\mid x)-f^{H}(\tilde q_\tau(x)\mid x) =o_p(1),\\[5pt]
& \nu_\tau(x):=\hat F^{H}(\tilde q_\tau(x)\mid x)-F^{H}(\tilde q_\tau(x)\mid x) = o_p(1).
\end{align*}
\end{assumption}
Assumption~\ref{ass:pilot_local} collects the local consistency conditions needed for the one-step expansion. 
The first condition, \(e_\tau(x)=o_p(1)\), requires the wrapper estimator to lie in a shrinking neighborhood of the target HF quantile. 
The second condition, \(\eta_\tau(x)=o_p(1)\), requires the estimated HF density to be locally consistent at the wrapper estimate, so that the reciprocal 
\(\hat f^H(\tilde q_\tau(x)\mid x)^{-1}\) is well behaved. 
The third condition, \(\nu_\tau(x)=o_p(1)\), requires the estimated HF CDF to be locally consistent at the same point. 
Under these conditions, we can expand \(F^H(\cdot\mid x)\) and \(f^H(\cdot\mid x)\) around \(q_\tau(x)\), and then apply a reciprocal expansion for 
\(\hat f^H(\tilde q_\tau(x)\mid x)^{-1}\). 
This yields the following first-order characterization of the corrected estimator.

\begin{theorem}
\label{thm:multiple_bias_nonparam}
Under Assumptions~\ref{ass:H_quantile_reg} and \ref{ass:pilot_local}, for each fixed \(x\) and \(\tau\),
\[
\tilde q_{\tau,\gamma}(x)-q_\tau(x)
=
(1-\gamma)e_\tau(x)
-\gamma\frac{\nu_\tau(x)}{f^{H}(q_\tau(x)\mid x)}
+R_{\tau}^{(2)}(x),
\]
where the remainder term 
\[
R_{\tau}^{(2)}(x)
=
O_p\!\bigl(
e_\tau(x)^2
+e_\tau(x)\eta_\tau(x)
+e_\tau(x)\nu_\tau(x)
+\nu_\tau(x)\eta_\tau(x)
\bigr).
\]
\end{theorem}

Theorem~\ref{thm:multiple_bias_nonparam} shows that the one-step estimator has two leading terms: the remaining wrapper contribution \((1-\gamma)e_\tau(x)\), and the local HF CDF estimation error \(\nu_\tau(x)/f^{H}(q_\tau(x)\mid x)\). The remaining terms are higher order. In particular, the local HF density estimation error \(\eta_\tau(x)\) appears only through its products with either the wrapper error \(e_\tau(x)\) or the local HF CDF error \(\nu_\tau(x)\), while the wrapper bias contributes the second-order term \(e_\tau(x)^2\). Hence, when \(\nu_\tau(x)\) and \(\eta_\tau(x)\) are of the same stochastic order and \(e_\tau(x)=o_p(1)\), accurate estimation of the HF density $f^H(\cdot\mid x)$ is \emph{less critical} than accurate estimation of the HF CDF $F^{H}(\cdot\mid x)$, because the density error enters only through interaction terms.
In the special case \(\gamma=1\), we have
\[
\tilde q_{\tau,1}(x)-q_\tau(x)
=
-\frac{\nu_\tau(x)}{f^{H}(q_\tau(x)\mid x)}
+ R_{\tau}^{(2)}(x).
\]
Thus, when the wrapper error remains non-negligible but the HF CDF can be estimated accurately, the one-step correction can substantially improve upon the wrapper \emph{without directly estimating the HF quantile function \(q_\tau(\cdot)\)}.

\subsection{One-step correction via cross-fitting}
\label{sect:one_step_cross}
We implement the one-step correction using \(K\)-fold cross-fitting on the HF sample. Let
\(
\mathcal I_1,\dots,\mathcal I_K
\)
be a partition of the HF indices \(\{1,\dots,n_H\}\). For each fold \(k\), we fit all HF-dependent quantities using only the training sample
\[
\mathcal D^{H,(-k)}:=\{(X_i^H,Y_i^H): i\notin \mathcal I_k\},
\]
while using the full LF sample \(\mathcal D^L\) to construct the LF CDF estimator \(\hat F^L\). This yields a cross-fitted wrapper estimator \(\tilde q_\tau^{(-k)}(\cdot)\), together with cross-fitted HF distribution estimators \(\hat F^{H,(-k)}(\cdot\mid \cdot)\) and \(\hat f^{H,(-k)}(\cdot\mid \cdot)\).

For each \(i\in \mathcal I_k\), we then form the cross-fitted one-step prediction
\[
\tilde q_{\tau,\gamma}^{(-k)}(X_i^H)
=
\tilde q_\tau^{(-k)}(X_i^H)
-
\gamma\,
\frac{\hat F^{H,(-k)}(\tilde q_\tau^{(-k)}(X_i^H)\mid X_i^H)-\tau}
{\hat f^{H,(-k)}(\tilde q_\tau^{(-k)}(X_i^H)\mid X_i^H)}.
\]
We choose the step size \(\gamma\) by minimizing the cross-fitted pinball loss,
\[
\hat\gamma
\in
\argmin_{\gamma\in\Gamma}
\sum_{k=1}^K \sum_{i\in \mathcal I_k}
\rho_\tau\!\bigl(Y_i^H-\tilde q_{\tau,\gamma}^{(-k)}(X_i^H)\bigr),
\]
where \(
\rho_\tau(z)=z\bigl(\tau-\one_{\{z<0\}}\bigr)
\) and
\(\Gamma\subset[0,1]\) is a user-specified grid. 
If \(\tilde q_\tau^{(-k)}(X_i^H)\) is already a good quantile estimate on the held-out HF data, then moving in the correction direction does not improve the pinball loss, and \(\hat\gamma\) is chosen close to \(0\). By contrast, if the wrapper has a systematic bias, then a larger correction reduces the held-out pinball loss, and \(\hat\gamma\) is chosen closer to \(1\). Intermediate values of \(\hat\gamma\) arise when the correction is useful but the full Newton step is too aggressive, so that a damped update gives the best out-of-sample quantile fit.

After selecting \(\hat\gamma\), we refit \(\tilde q_\tau(\cdot)\), \(\hat F^H(\cdot\mid\cdot)\), and \(\hat f^H(\cdot\mid\cdot)\) on the full HF sample and define the final estimator as
\(
\tilde q_{\tau,\hat\gamma}(x)
\)
in \eqref{equ:q_tau_gamma}. Cross-fitting helps prevent overfitting because both the wrapper estimator and the local HF correction are constructed from the same scarce HF sample. It avoids tuning the correction on the same observations used to construct the wrapper estimator.

\subsection{Multi-step correction beyond the local quantile link}
\label{sect:ms}

The wrapper estimator is motivated by Assumption~\ref{assume:qq}, which requires the target HF quantile function to be representable as an LF quantile at some level. When this assumption fails, the wrapper may retain a nonvanishing approximation error for any sample size, so a single correction step need not be sufficient for consistency. In that case, it is natural to iterate the one-step update and directly target the HF quantile equation.

Starting from $q_\tau^{(0)}(x):=\tilde q_\tau(x),$
define the iterates
\[
q_\tau^{(m+1)}(x)
=
q_\tau^{(m)}(x)
-
\frac{\hat F^H(q_\tau^{(m)}(x)\mid x)-\tau}
{\hat f^H(q_\tau^{(m)}(x)\mid x)},
\qquad m=0,1,2,\dots.
\]
This is Newton's method applied to the empirical HF estimating equation
\[
\Psi_n(q;x):=\hat F^H(q\mid x)-\tau.
\]

\begin{assumption}
\label{ass:multistep}
Fix \(x\in\mathcal X\) and \(\tau\in(0,1)\). With probability tending to one, there exist an open interval \(I_x\subset\mathbb R\) containing a root \(\hat q_\tau^H(x)\) of \(\Psi_n(\cdot;x)\), and constants \(c,L,\delta_0>0\), such that:
\begin{enumerate}[label=(\roman*)]
    \item \(\Psi_n(\cdot;x)\in C^2(I_x)\), and \(\hat f^H(q\mid x)=\partial_q \hat F^H(q\mid x)=\Psi_n'(q;x)\) for all \(q\in I_x\);
    \item \(0<c\le \hat f^H(q\mid x)\) and \(|\partial_q \hat f^H(q\mid x)|\le L\) for all \(q\in I_x\);
    \item \([\hat q_\tau^H(x)-\delta_0,\ \hat q_\tau^H(x)+\delta_0]\subset I_x\);
    \item \(|q_\tau^{(0)}(x)-\hat q_\tau^H(x)|\le \delta_0\) and \(\delta_0L/(2c)<1\).
\end{enumerate}
\end{assumption}

Assumption~\ref{ass:multistep} is mild when \(\hat F^H(\cdot\mid x)\) is a smoothed estimator and \(\hat f^H(\cdot\mid x)\) is taken as its derivative. Indeed, the lower bound in (ii) implies that \(\hat F^H(\cdot\mid x)\) is strictly increasing there, so the equation \(\hat F^H(q\mid x)=\tau\) can have at most one solution on \(I_x\). Thus, uniqueness of \(\hat q_\tau^H(x)\) follows automatically from part (ii), together with existence of a root in \(I_x\). This assumption would be more restrictive only if \(\hat F^H\) were taken to be a step-function empirical CDF, in which case differentiability fails and the solution set need not be unique.

\begin{proposition}
\label{prop:iterated_onestep_main}
Under Assumption~\ref{ass:multistep}, with probability tending to one, the iterates \(q_\tau^{(m)}(x)\) remain in \(I_x\) and converge to \(\hat q_\tau^H(x)\) as \(m\to\infty\).
\end{proposition}

If the estimated HF CDF \(\hat F^H(\cdot\mid x)\) is consistent on \(I_x\) and the true HF density \(f^H(\cdot\mid x)\) remains locally positive, then the empirical root \(\hat q_\tau^H(x)\) is itself consistent for the true quantile \(q_\tau(x)\). Since the iterates converge to this empirical root, repeated one-step correction can therefore target the HF quantile. Thus, even when Assumption~\ref{assume:qq} does not hold exactly, iterative correction can still recover consistency provided the empirical HF root is consistent and the initialization lies in the Newton basin described above.

Multi-step correction follows a different update trajectory from the mixed estimator discussed in \Cref{sect:corrections}.
The mixed estimator interpolates between two fixed estimators and may remain inaccurate when both endpoints are inaccurate.
By contrast, multi-step correction successively updates the estimate using the HF quantile equation and can move toward a more accurate HF quantile estimate. 
In implementation, multi-step correction is cross-fitted in the same way as the one-step correction in \Cref{sect:one_step_cross}: for each fold, the wrapper and HF correction CDF are fit without using the held-out HF observations, and the iterated updates are evaluated on those held-out observations.
In practice, the number of correction steps can be selected by the same cross-fitted validation strategy used in the previous subsection.

\section{Experiments}
\label{sec:experiments}

This section evaluates the wrapper estimator and its corrected version on synthetic datasets and five scientific datasets. The code for reproducing the experiments is available at
\href{https://github.com/Yixiang-Kenzo/Multi-Fidelity-Quantile-Regression}{https://github.com/Yixiang-Kenzo/MFQR}.
We first describe competing methods and performance metrics below.

\subsection{Methods \& Metrics}
We implement all HF and LF conditional mean, density, CDF, and quantile estimators using random forests with the same hyperparameters across methods.
The random-forest density, CDF, and quantile construction is reviewed in Appendix~\ref{app:implementation}.
In \Cref{app:gp_results}, we repeat the experiments with Gaussian process (GP) regression models to verify that the results are not specific to the random-forest implementation.

Throughout most experiments, we compare five main methods spanning no transfer, mean transfer, feature transfer, and wrapper-based transfer:
\begin{itemize}[leftmargin=*]
    \item \textbf{HF-Only}: uses only the HF data to estimate the HF conditional quantiles, corresponding to the direct estimator \(\hat q_\tau^H(x)\) discussed above.

    \item \textbf{Tr-Mean}: fits a conditional mean model \(\hat\mu_L\) on the LF data \(\mathcal D^{L}\), and then fits a linear function \(\hat a+\hat b\hat\mu_L\) on the HF data \(\mathcal D^{H}\) to estimate the HF mean model \(\hat\mu_H\). 
    This implements the residual learning idea in the classical autoregressive multi-fidelity framework \citep{kennedy2000predicting,perdikaris2017nonlinear}. 
    We compute the HF residuals \(Y_i^H-\hat\mu_H(X_i^H)\), estimate their conditional CDF using \(\mathcal D^{H}\), invert the fitted CDF to obtain residual quantiles, and then add back the mean prediction \(\hat\mu_H\).

    \item \textbf{Tr-Augment}: augments each HF covariate vector \(X_i^H\) with the LF mean prediction \(\hat\mu_L(X_i^H)\). 
    The HF conditional CDF is then fitted on these augmented covariates. 
    This follows the feature augmentation strategy commonly used in multi-fidelity learning and transfer learning.

    \item \textbf{MFQR}: uses the wrapper estimator \(\tilde q_\tau(x)\) implemented in Algorithm~\ref{alg:wrapper}. 

    \item \textbf{MFQR+OS}: uses the one-step corrected wrapper estimator \(\tilde q_{\tau,\hat\gamma}(x)\) in \eqref{equ:q_tau_gamma}, where the step size \(\gamma\) is chosen by cross-fitting as described in \Cref{sect:one_step_cross}.
\end{itemize}
In the misinformative LF data regime, we additionally evaluate MFQR+MS, which applies the multi-step correction described in \Cref{sect:ms}.
We assess all methods by their estimates of the conditional quantiles at levels \(\tau=0.05\) and \(\tau=0.95\). 
In the synthetic experiments, where the true HF conditional quantiles are known, we report the average empirical mean squared error (MSE) of the two estimated quantiles on the test set.

For a given method, let \(\hat q_{0.05}\) and \(\hat q_{0.95}\) denote its fitted lower and upper HF conditional quantile functions. 
For both synthetic and scientific datasets, we apply split conformalized quantile regression (CQR) \citep{romano2019conformalized} to turn these fitted quantiles into marginally valid prediction intervals. 
On a held-out HF calibration set \(\{(X_i^{\mathrm{cal}},Y_i^{H,\mathrm{cal}})\}_{i=1}^{n_{\mathrm{cal}}}\), we compute 
\[
s_i
=
\max\left\{
\hat q_{0.05}(X_i^{\mathrm{cal}})-Y_i^{H,\mathrm{cal}},
\;
Y_i^{H,\mathrm{cal}}-\hat q_{0.95}(X_i^{\mathrm{cal}})
\right\}.
\]
For coverage level \(1-\alpha=0.9\), let \(\hat q_{\mathrm{cal}}\) be the 
\(\lceil (1-\alpha)(n_{\mathrm{cal}}+1)\rceil\)-th smallest $s_i$. 
The final interval is
\(
\left[
\hat q_{0.05}(x)-\hat q_{\mathrm{cal}},
\;
\hat q_{0.95}(x)+\hat q_{\mathrm{cal}}
\right].
\)
Since all methods are calibrated to the same level, we compare their average interval widths.

\subsection{Synthetic data}
\label{sec:synthetic}

Following the structural model~\eqref{eq:two_model}, we simulate three data regimes. 
The informative regime has similar LF and HF distributional shapes.
The non-informative regime introduces LF heteroscedasticity absent from the HF data. 
The misinformative regime has a nonlinear LF--HF mean discrepancy, so direct HF estimation can outperform the uncorrected wrapper.

\subsubsection{Informative LF regime}
\label{sec:informative}

In the informative regime, the HF and LF data share the same covariate distribution, $X\sim \mathrm{Uniform}(1.5,4.5),$ the same mean function
$\mu(x)=0.5x^2-2x+1$,
and baseline scale $\sigma(x)=0.1+0.35\sin^2(3\pi x).$
We generate
\[
Y^H=\mu(x)+\sigma(x)W^H,\qquad
Y^L=\mu(x)+1.7\,\sigma(x)W^L,
\]
where \(W^H\sim\mathcal N(0,1)\) and \(W^L\) is a standardized Student-\(t(10)\) random variable. Thus, the scale distortion in \eqref{eq:two_model} is constant, \(\rho(x)=1.7\). 
As shown in \Cref{fig:informative_data}, the simulated LF and HF conditional distributions have similar periodic shapes induced by the scale function \(\sigma(x)\). 
But because of the larger scale distortion and heavier-tailed LF noise, the LF responses spread more widely around the slowly increasing mean curve than the HF responses.
In total, we generate \(N=5000\) covariates \(X_i\) and split them into four disjoint subsets: \(n_L=1250\) LF training observations, \(n_H=125\) HF training observations, \(n_{\mathrm{cal}}=250\) calibration observations, and \(n_{\mathrm{test}}=3375\) test observations.

Table~\ref{tab:informative} reports the quantile MSE and conformal interval results on the test set. 
HF-Only and Tr-Augment have larger MSEs and wider intervals than the other methods. 
Tr-Mean substantially reduces both MSE and interval width by transferring the LF mean structure. 
MFQR and MFQR~+~OS improve further by transferring information through the LF conditional distribution.

\begin{figure}[t]
\vspace{-10pt}
    \centering
    \includegraphics[width=0.95\textwidth]{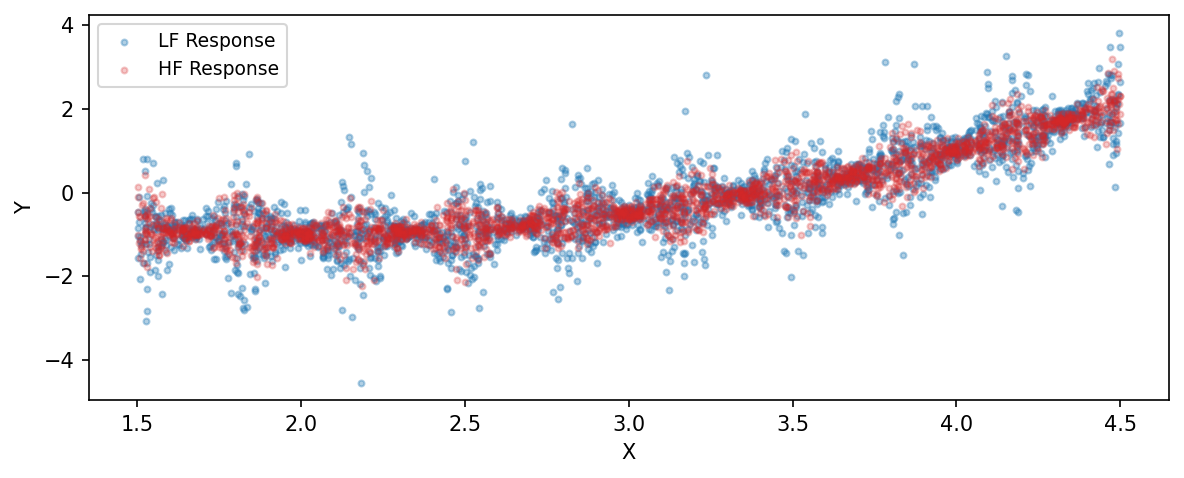}
 \caption{Illustration of LF and HF data in the informative LF regime;  the sample sizes for the experiments are specified in the text.}
    \label{fig:informative_data}
\end{figure}

\begin{table}[t]
\centering
\caption{Performance in the informative LF regime: quantile MSE, empirical coverage, and average interval width on the test set.}
\label{tab:informative}
\small
\begin{tabular}{l@{\hskip 34pt}c@{\hskip 34pt}c@{\hskip 34pt}c}
\toprule
Method & $\mathrm{MSE}$  & Coverage \ (\%) & Width \\
\midrule
    HF-Only    & 0.186  & 91.7 & 1.401\\
Tr-Mean    & 0.067  & 92.6 & 1.260\\
Tr-Augment & 0.183  & 91.2 & 1.370\\
MFQR       & 0.050  & 90.5 & 1.103 \\
MFQR+OS  & \textbf{0.039}  & 90.5 & \textbf{1.061} \\
\bottomrule
\end{tabular}
\end{table}

\clearpage

\begin{figure}[H]
\vspace{-35pt}
    \centering
    \begin{subfigure}[t]{0.95\textwidth}
        \centering
        \includegraphics[width=\textwidth]{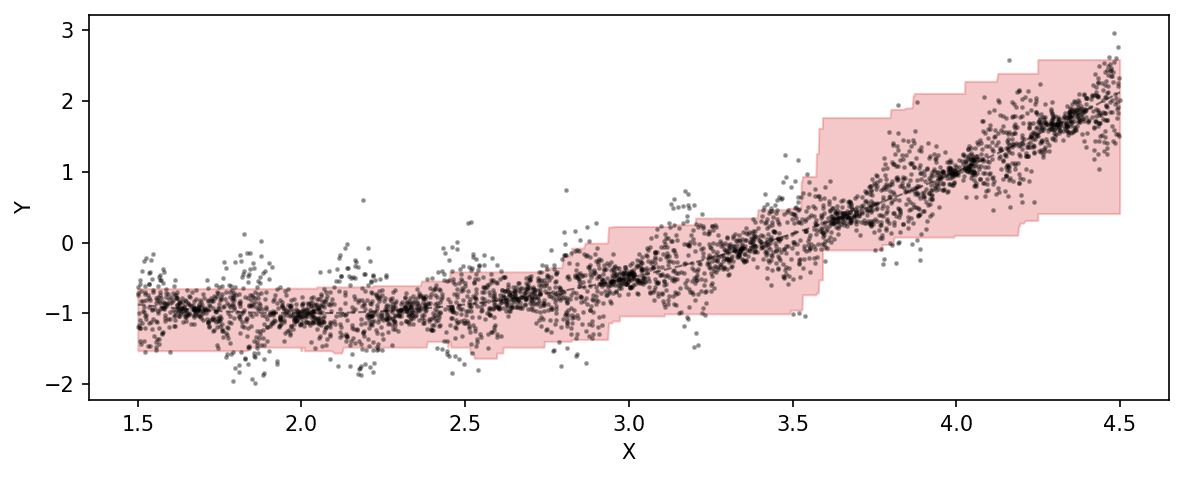}
        \caption{HF-Only: direct quantile regression on HF data.}
    \end{subfigure}
    
    \vspace{0.5em}
    
    \begin{subfigure}[t]{0.95\textwidth}
        \centering
        \includegraphics[width=\textwidth]{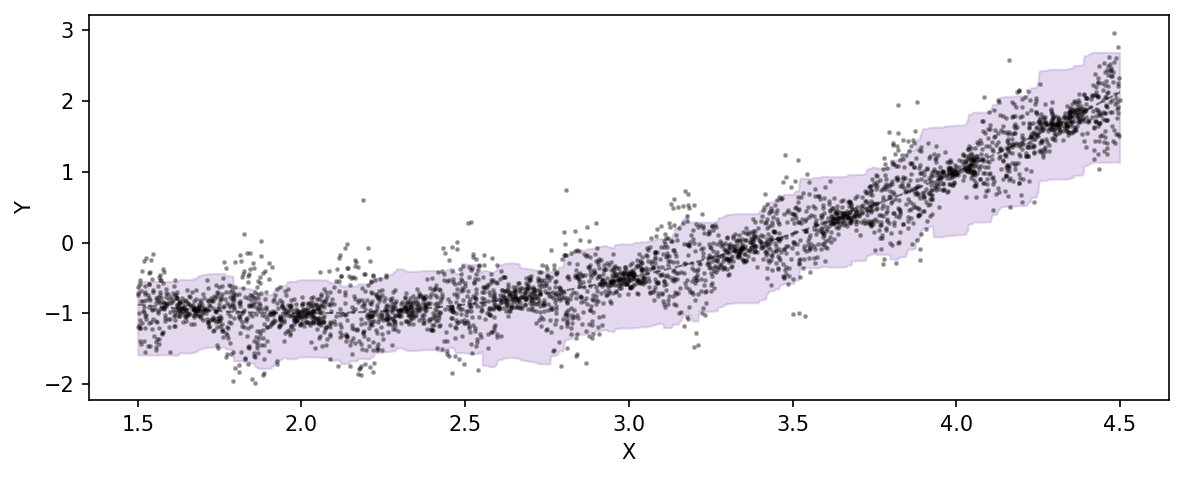}
        \caption{Tr-Mean: quantile regression after mean transfer.}
    \end{subfigure}
    
    \vspace{0.5em}
    
    \begin{subfigure}[t]{0.95\textwidth}
        \centering
        \includegraphics[width=\textwidth]{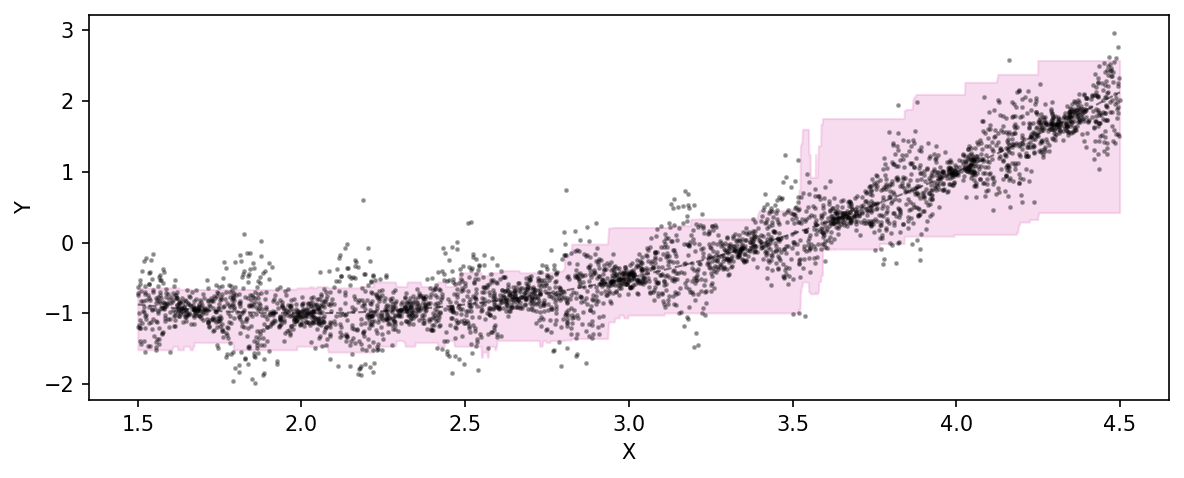}
        \caption{Tr-Augment: quantile regression with augmented features.}
    \end{subfigure}
    \caption{Conformal prediction intervals based on quantile estimates from three baseline methods in the informative LF regime.}
    \label{fig:informative_cqr_baselines}
\end{figure}

\begin{figure}[H]
\vspace{-35pt}
    \centering
    \begin{subfigure}[t]{0.95\textwidth}
        \centering
        \includegraphics[width=\textwidth]{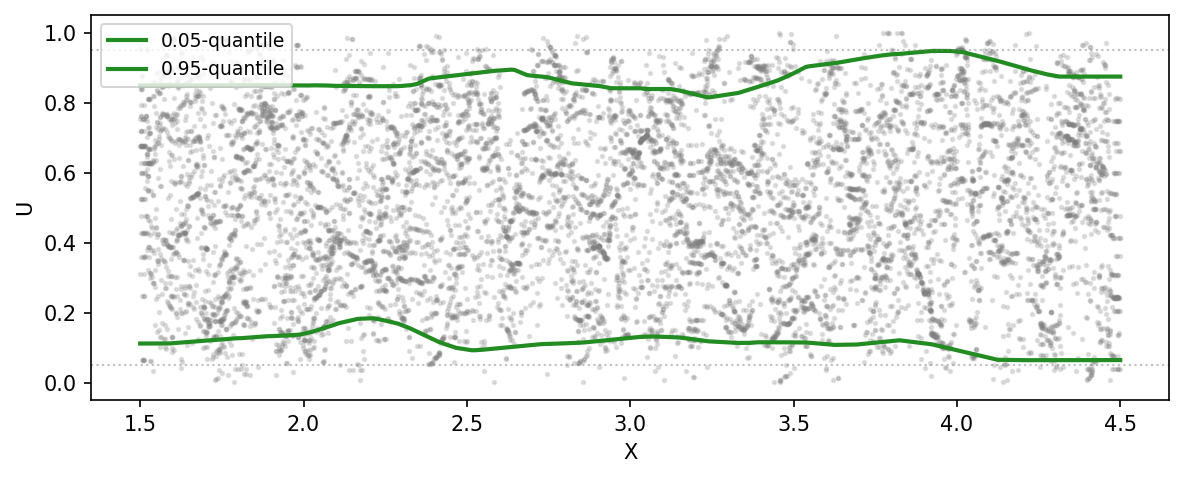}
        \caption{Wrapped pseudo-responses and their quantiles.}
    \end{subfigure}

    \vspace{0.5em}

    \begin{subfigure}[t]{0.95\textwidth}
        \centering
        \includegraphics[width=\textwidth]{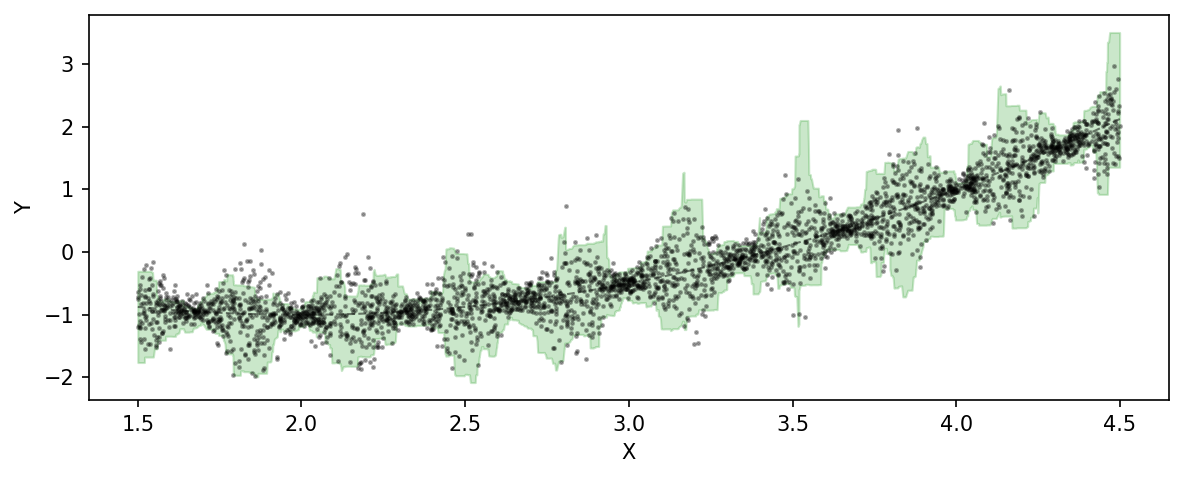}
        \caption{MFQR: wrapper estimator using the LF conditional distribution.}
    \end{subfigure}

    \vspace{0.5em}

    \begin{subfigure}[t]{0.95\textwidth}
        \centering
        \includegraphics[width=\textwidth]{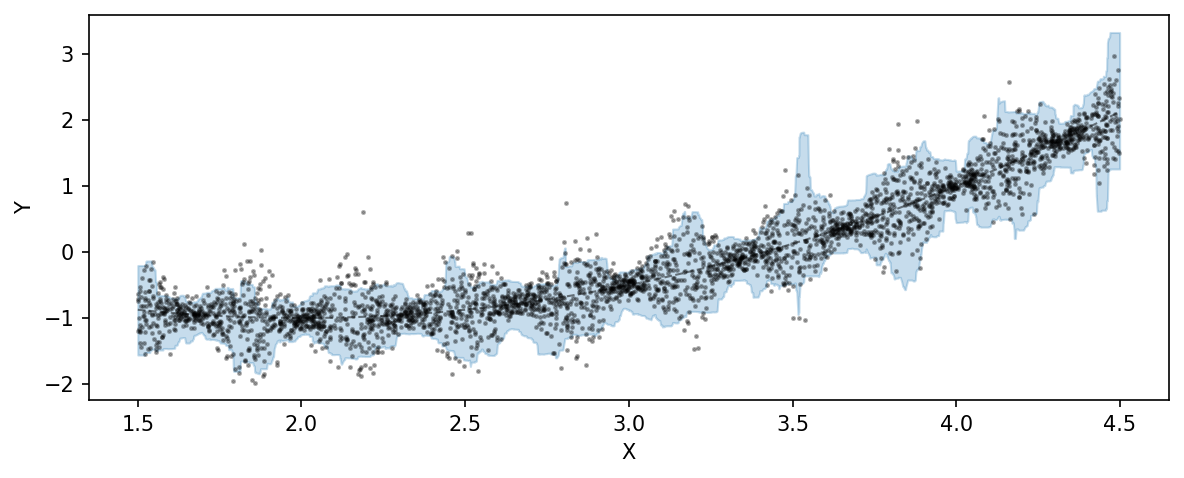}
        \caption{MFQR with one-step correction (MFQR+OS).}
    \end{subfigure}
    \caption{MFQR results in the informative LF regime. The level functions \(u_{0.05}(x)\) and \(u_{0.95}(x)\), defined as conditional quantiles of the wrapped pseudo-responses, are nearly flat. The resulting MFQR and MFQR~+~OS conformal prediction intervals closely track the HF response variation over \(X\).}
    \label{fig:informative_cqr_mfqr}
    \vspace{10pt}
\end{figure}

\clearpage

\Cref{fig:informative_cqr_baselines} visualizes the conformal prediction intervals for the baseline methods. 
The intervals from HF-Only and Tr-Augment fail to  capture the periodic scale variation of the responses over \(X\). 
This suggests that the augmented LF feature does not provide much additional information beyond the raw covariate \(X\) for predicting the HF response. 
Tr-Mean produces narrower intervals but still struggles to capture the periodic scale.

\Cref{fig:informative_cqr_mfqr} shows the wrapped pseudo-responses and the conformal prediction intervals from MFQR and MFQR~+~OS. 
The level functions \(u_{0.05}(x)\) and \(u_{0.95}(x)\), defined as conditional quantiles of the wrapped pseudo-responses, are nearly flat over \(X\). 
This agrees with Proposition~\ref{prop:wrapper_smoother_extended}:  under model~\eqref{eq:two_model},
$u_\tau(x)=G_L\!\left(G_H^{-1}(\tau)/1.7\right),$
which is constant in \(x\). 
This constant level function is maximally smooth,
while the original HF quantile still inherits variation from the polynomial mean and the oscillating scale function \(\sin^2(3\pi x)\). 
Thus, the wrapper converts the original nonlinear HF quantile estimation problem into the simpler task of estimating an approximately flat level function, allowing MFQR and MFQR~+~OS to produce accurate quantile estimates and sharp prediction intervals even with limited HF data.

\subsubsection{Non-informative LF regime}
\label{sec:noninformative}

The HF and LF data share the same covariate distribution,
\(X\sim \mathrm{Uniform}(0,0.5)\), and the same mean function 
$\mu(x)=0.6[2\sin(2\pi x)+x].$
We generate
\[
Y^H=\mu(x)+\sigma_H(x)\,W^H,\qquad
Y^L=\mu(x)+\sigma_L(x)W^L,
\]
where \(W^H\sim\mathcal N(0,1)\), \(W^L\) is a standardized Student-\(t(10)\), $\sigma_H(x)=0.1$, and $\sigma_L(x)=0.1+6.4(x-0.25)^2$.
As shown in \Cref{fig:noninformative_data_transform}, the HF responses have constant scale, whereas the LF responses exhibit clear heteroscedasticity.
In total, we generate \(N=2500\) covariates and split them into four disjoint subsets: \(n_L=1250\) LF training observations, \(n_H=125\) HF training observations, \(n_{\mathrm{cal}}=250\) calibration observations, and \(n_{\mathrm{test}}=875\) test observations.

\begin{figure}[t]
    \centering
        \includegraphics[width=\textwidth]{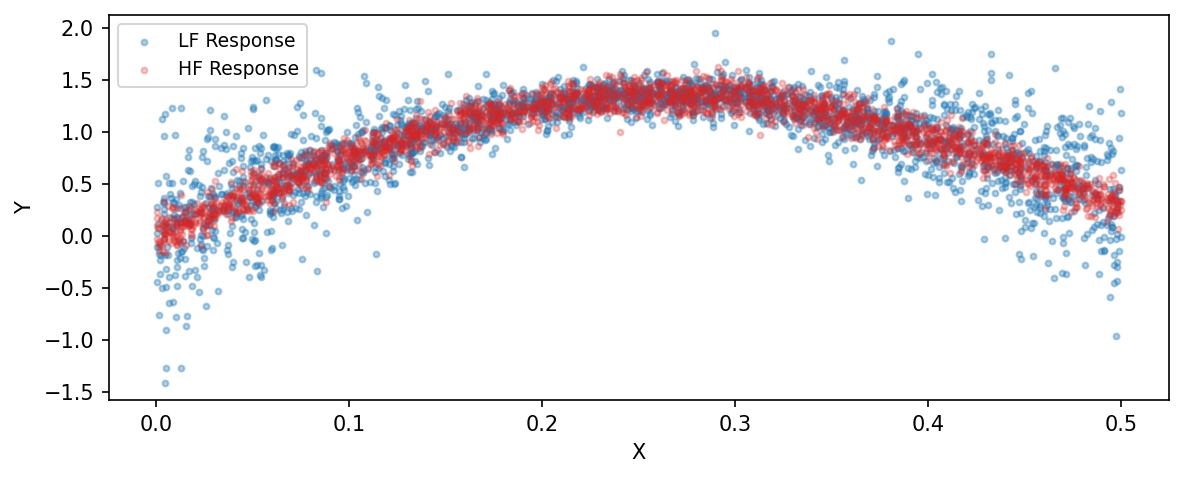}
 \caption{Illustration of LF and HF data in the non-informative LF regime;  the sample sizes for the experiments are specified in the text.}
    \label{fig:noninformative_data_transform}
\end{figure}

\begin{table}[t]
\centering
\caption{Performance in the non-informative LF regime: quantile MSE, empirical coverage, and average interval width on the test set.}
\label{tab:noninformative}
\small
\begin{tabular}{l@{\hskip 34pt}c@{\hskip 34pt}c@{\hskip 34pt}c}
\toprule
Method & $\mathrm{MSE}$ & Coverage\ (\%) & Width  \\
\midrule
HF-Only & 0.038 & 92.2  & 0.558\\
MFQR & 0.025  & 91.8 & 0.501 \\
MFQR+OS & \textbf{0.008} & 92.6  & \textbf{0.418}\\
\bottomrule
\end{tabular}
\end{table}

In this non-informative LF regime, \Cref{tab:noninformative} shows that, by transferring the shared mean function, MFQR alone yields only a modest improvement over HF-Only, with roughly a \(35\%\) reduction in quantile MSE.
By contrast, MFQR~+~OS reduces the quantile MSE by around \(78\%\) and the interval width by nearly \(25\%\) relative to HF-Only. This is exactly the regime anticipated by \Cref{thm:multiple_bias_nonparam}: because the HF conditional distribution has a simple constant-scale
structure, the HF CDF \(\hat F^H(\cdot\mid x)\) can be estimated
accurately even from the limited HF sample, so the local CDF error
\(\nu_\tau(x)\) is small. The one-step correction therefore improves MFQR by replacing a non-negligible wrapper bias \(e_\tau(x)\) with a much smaller HF estimation error.

\Cref{fig:noninformative_cqr} visualizes the prediction intervals from HF-Only, MFQR, and MFQR~+~OS.
HF-Only struggles to learn the nonlinear response surface from limited HF data. 
MFQR captures the main response shape, but its intervals vary irregularly across \(X\), because the wrapped pseudo-responses inherit heteroscedastic patterns from the LF data.
The one-step correction in MFQR~+~OS mitigates this issue without re-estimating the HF quantile surface. 
The resulting intervals are tighter and smoother, inheriting less of the spurious LF heteroscedasticity across most values of \(X\).

\begin{figure}[H]
\vspace{-35pt}
    \centering
    \begin{subfigure}[t]{0.95\textwidth}
        \centering
        \includegraphics[width=\textwidth]{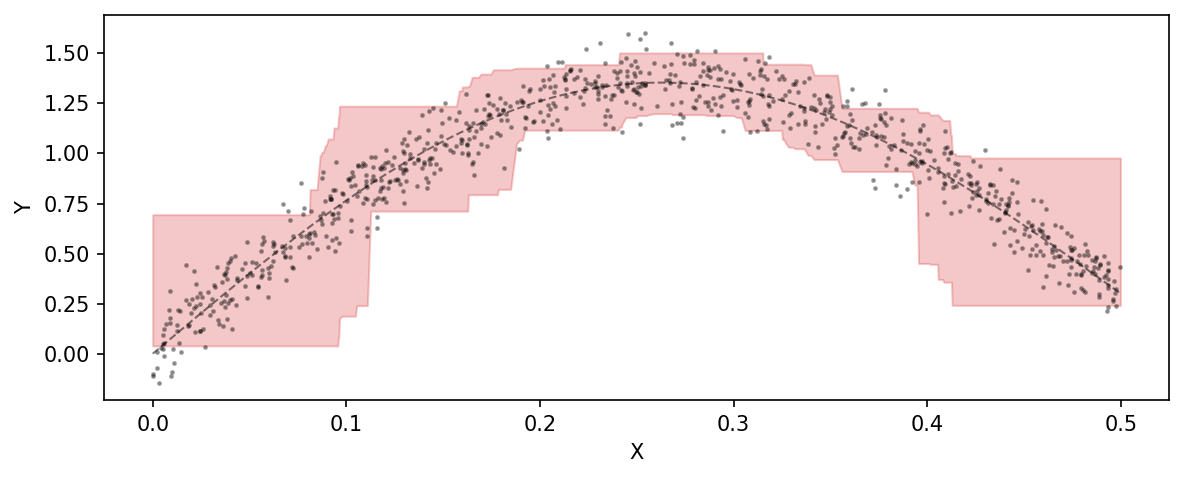}
        \caption{HF-Only: direct quantile regression on HF data.}
    \end{subfigure}
    
    \vspace{0.5em}
    
    \begin{subfigure}[t]{0.95\textwidth}
        \centering
        \includegraphics[width=\textwidth]{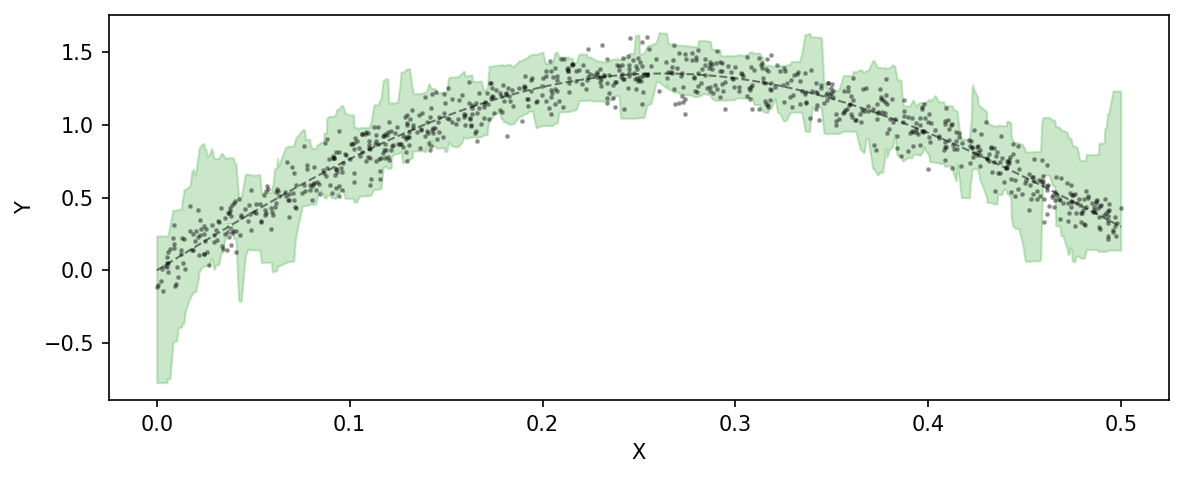}
        \caption{MFQR: wrapper estimator using the LF conditional distribution.}
    \end{subfigure}
    
    \vspace{0.5em}
    
    \begin{subfigure}[t]{0.95\textwidth}
        \centering
        \includegraphics[width=\textwidth]{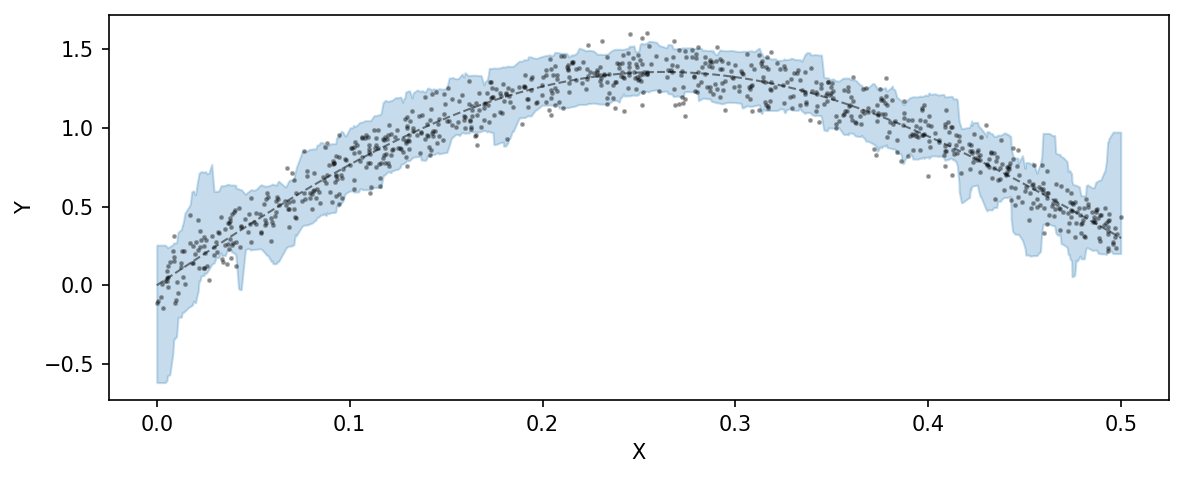}
        \caption{MFQR with one-step correction (MFQR+OS).}
    \end{subfigure}
    \caption{Conformal prediction intervals in the non-informative LF regime. MFQR yields more informative intervals than HF-Only over \(X\), while the one-step correction further shrinks and smooths the MFQR intervals. These improvements are achieved while maintaining coverage near the target level.}
    \label{fig:noninformative_cqr}
\end{figure}

\subsubsection{Misinformative LF regime}
\label{sec:mean_failure}

We next simulate HF and LF data with the same covariate distribution, \(X\sim \mathrm{Uniform}(0,1)\), and the same scale function
$\sigma(x)=0.15+0.3(2x-1)^2$,
but with very different mean functions:
\[
Y^H=\sin(2\pi x)+\sigma(x)W^H,
\qquad
Y^L = 0.5\sin(2\pi x) + 1.5\sin(4\pi x) + \sigma(x) W^L,
\]
where the noise variables \(W^H,W^L\stackrel{\mathrm{iid}}{\sim}\mathcal N(0,1)\).
As shown in \Cref{fig:misinformative}, the LF and HF responses have different shapes over \(X\). 
Because the conditional distributions are Gaussian and supported on the real line, a formal local quantile link still exists. 
However, this link is not useful for transfer. 
The induced level function \(u_\tau(x)\) must absorb the nonlinear mismatch between the LF and HF mean functions, and this mismatch cannot be removed by a simple affine transformation. 

We generate \(N=5000\) observations, with \(n_L=2500\) for LF training, \(n_H=350\) for HF training, \(n_{\mathrm{cal}}=500\) for calibration, and the remaining observations for testing. 
In this setting, direct HF estimation is more reliable than the wrapper estimator. 
As shown in \Cref{fig:misinformative_cqr_part1}, HF-Only produces relatively tight prediction intervals while maintaining valid coverage.

Table~\ref{tab:iter_meanfail} reports the quantile MSE, empirical coverage, and interval width in the misinformative regime. 
HF-Only outperforms MFQR because the wrapper makes the estimation problem more complex: as shown in \Cref{fig:misinformative}, the HF response curve is smoother than the LF response curve. 
The one-step correction substantially reduces the MSE by moving the wrapper estimate toward the HF quantile estimating equation, but a single step may be insufficient when the local quantile link is not useful for transfer. 
In this case, as discussed in \Cref{sect:ms}, multi-step correction further reduces the error, and MFQR+MS achieves an MSE comparable to HF-Only. 
\Cref{fig:misinformative_cqr_part2} illustrates this improvement: the corrected intervals are much tighter than those from the uncorrected wrapper.

Taken together, the three synthetic experiments show that the proposed framework remains robust across informative, non-informative, and misinformative LF regimes. The correction steps are crucial because they adapt the wrapper estimator toward the HF quantile equation when the LF data are unreliable.

\clearpage

\begin{figure}[H]
\vspace{-25pt}
    \centering
    \includegraphics[width=0.95\textwidth]{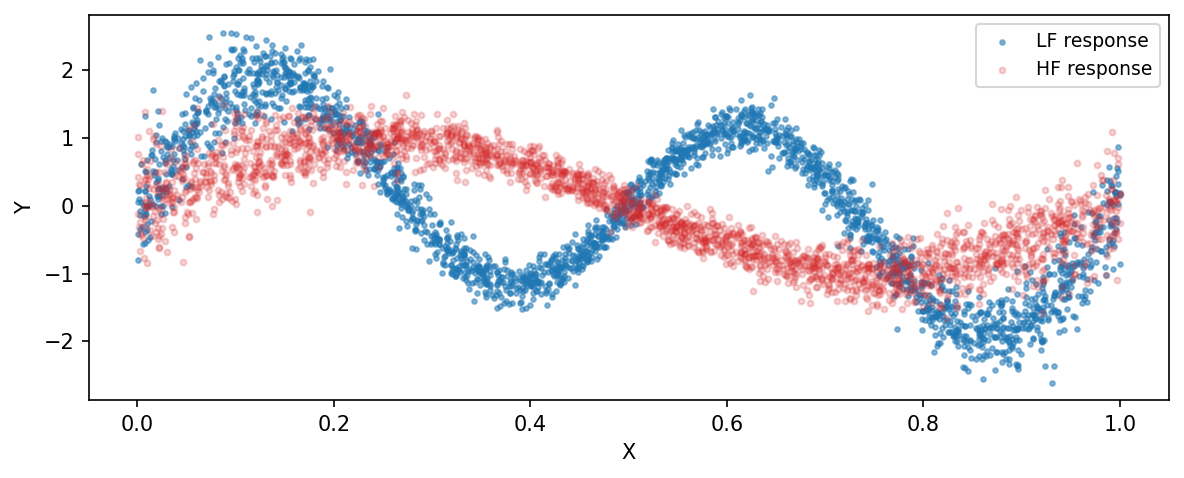}
 \caption{Illustration of LF and HF data in the misinformative LF regime;  the sample sizes for the experiments are specified in the text.}
   \label{fig:misinformative}
\end{figure}

\begin{figure}[H]
        \centering
        \includegraphics[width=0.95\textwidth]{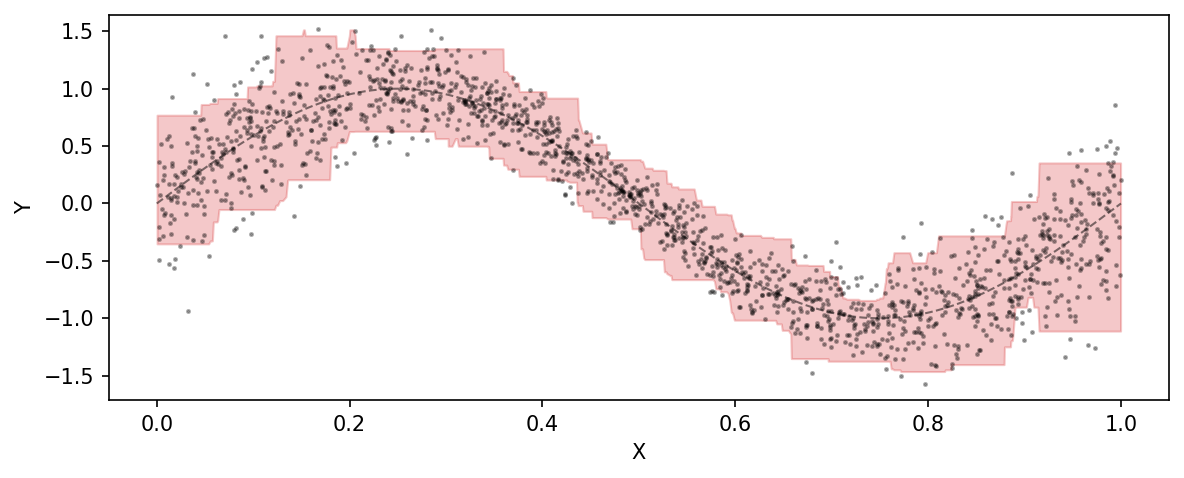}
    \caption{Prediction intervals of HF-Only in the misinformative LF regime.}
    \label{fig:misinformative_cqr_part1}
\end{figure}

\begin{table}[H]
\centering
\caption{Performance in the misinformative LF regime: quantile MSE, empirical coverage, and average interval width on the test set. MS denotes multi-step correction; cross-fitting selects \(m=5\) correction steps.}
\label{tab:iter_meanfail}
\centering
\small
\setlength{\tabcolsep}{12pt}
\begin{tabular}{l@{\hskip 34pt}c@{\hskip 34pt}c@{\hskip 34pt}c}
\toprule
Method &  $\mathrm{MSE}$ & Coverage (\%) & Width \\
\midrule
HF-Only  & 0.016 & 92.6 & 0.926 \\
MFQR  & 0.473 & 89.9 & 2.151 \\
MFQR+OS  & 0.054 & 89.5 & 1.214 \\
MFQR+MS  & \textbf{0.015} & 90.2 & \textbf{0.862} \\
\bottomrule
\end{tabular}
\end{table}

\clearpage

\clearpage

\clearpage
\begin{figure}[H]
\vspace{-35pt}
    \centering
 \begin{subfigure}[t]{0.95\textwidth}
        \centering
        \includegraphics[width=\textwidth]{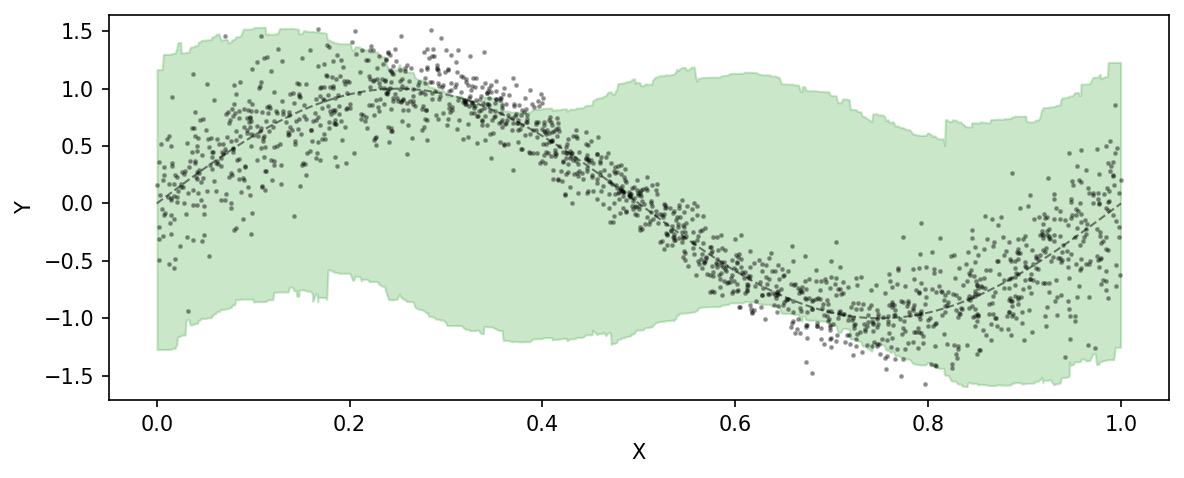}
        \caption{MFQR: wrapper estimator using the LF conditional distribution.}
        \label{fig:meanfail_cqr_mfqr}
    \end{subfigure}
    
    \vspace{0.5em}
    
    \begin{subfigure}[t]{0.95\textwidth}
        \centering
        \includegraphics[width=\textwidth]{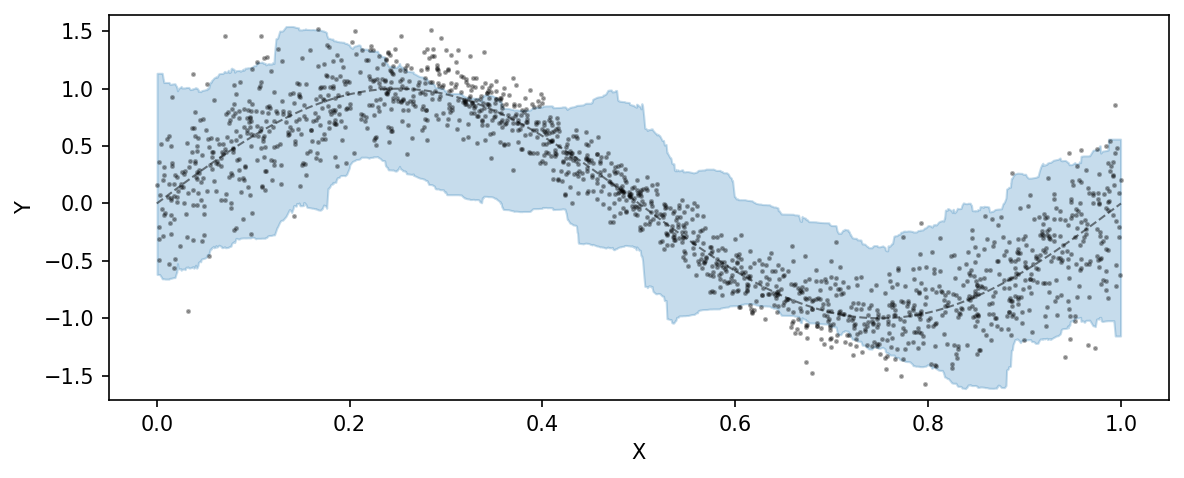}
        \caption{MFQR with one-step correction (MFQR+OS).}
        \label{fig:meanfail_cqr_os}
    \end{subfigure}
    
    \vspace{0.5em}
    
    \begin{subfigure}[t]{0.95\textwidth}
        \centering
        \includegraphics[width=\textwidth]{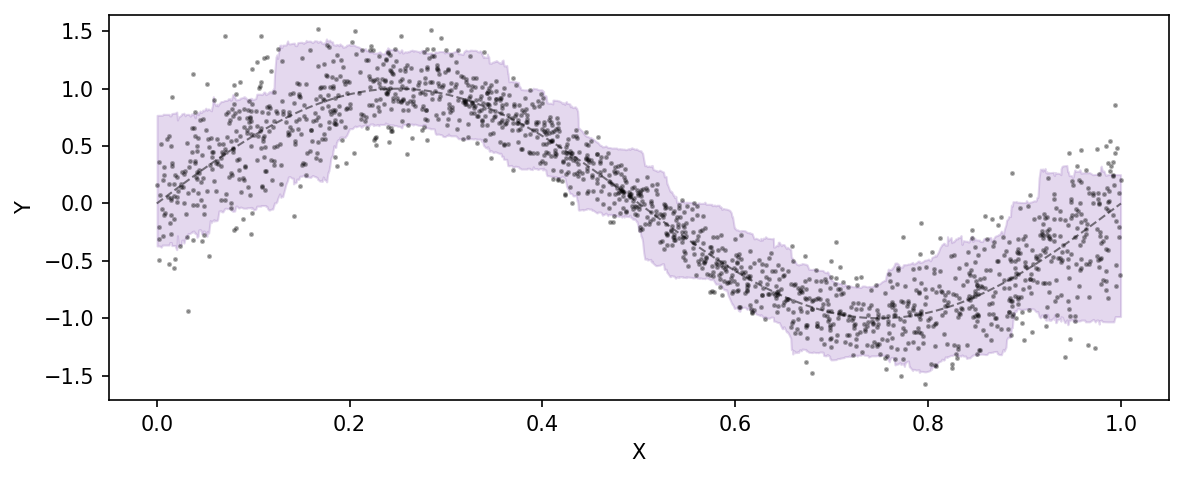}
        \caption{MFQR with multi-step correction (MFQR+MS).}
        \label{fig:meanfail_cqr_ms}
    \end{subfigure}

\caption{MFQR results in the misinformative LF regime. One-step and multi-step corrections enable MFQR to generate tighter prediction intervals.}
    \label{fig:misinformative_cqr_part2}
\end{figure}
\clearpage

\subsection{Scientific data}
\label{sec:scientific}

We evaluate all methods on five datasets spanning three domains. 
The first three are drawn from QeMFi, a multi-fidelity dataset of molecules for quantum chemical property prediction \citep{vinod2025qemfi}. 
Each QeMFi dataset contains \(N=15000\) molecular structures, where each
structure records a different three-dimensional arrangement of the atoms. Here
\(d\) denotes the dimension of the descriptor vector used as covariates;
molecules with more atoms typically have larger \(d\).

\begin{itemize}[leftmargin=*]
    \item \textbf{Acrolein} ($d{=}24$): 
    Acrolein is a small organic molecule with formula 
    C\(_3\)H\(_4\)O. 
    The dataset contains many molecular geometries of acrolein, represented by molecular descriptors, to predict the corresponding total energy. 
    The LF energies are computed with the minimal STO-3G basis set, which is computationally cheap but less accurate, while the HF energies are computed with the larger def2-TZVP basis set. 
The larger basis set yields more accurate energies at substantially higher computational cost.
    \item \textbf{Thymine} ($d{=}45$): Thymine is a larger organic molecule with formula 
    C\(_5\)H\(_6\)N\(_2\)O\(_2\). 
We use the same prediction task based on LF/HF basis-set pairing as in Acrolein: STO-3G energies are treated as LF responses, and def2-TZVP energies are treated as HF responses.  
\item \textbf{o-HBDI} (\(d{=}66\)): 
o-HBDI is a photoactive organic molecule with formula C\(_8\)H\(_8\)N\(_2\)O\(_2\). STO-3G energies are treated as LF responses, and def2-TZVP energies are treated as HF responses. 
\item \textbf{Burgers} (\(N=5000\), \(d=6\)): a multi-fidelity benchmark
based on velocity solutions to the viscous Burgers equation
\citep{burgers1948mathematical,perdikaris2017nonlinear,meng2020composite}.
The covariates describe the initial condition, viscosity, time, and location.
LF and HF responses are computed on coarse and fine meshes with 16 and 128 grid
points, respectively.

\item \textbf{Formation Energy (F-Energy)} (\(N=2500\), \(d=11\)): a crystal formation-energy dataset from the Materials Project \citep{jain2013materials}. Each observation has compositional and structural descriptors, with LF responses based on PBE and HF responses based on a more accurate density-functional approximation.
\end{itemize}

\begin{figure}[t]
\vspace{-35pt}
    \centering
    \begin{subfigure}[t]{0.9\textwidth}
        \centering
        \includegraphics[width=\textwidth]{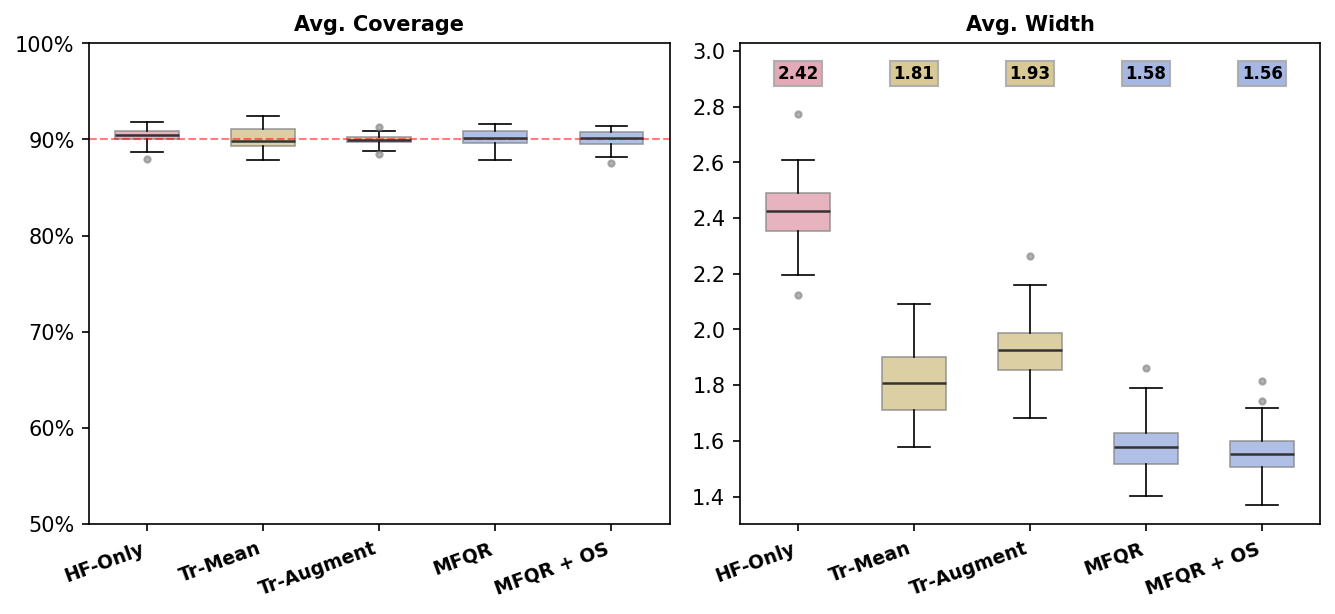}
        \caption{Acrolein.}
    \end{subfigure}
    
    \vspace{0.5em}
    
    \begin{subfigure}[t]{0.9\textwidth}
        \centering
        \includegraphics[width=\textwidth]{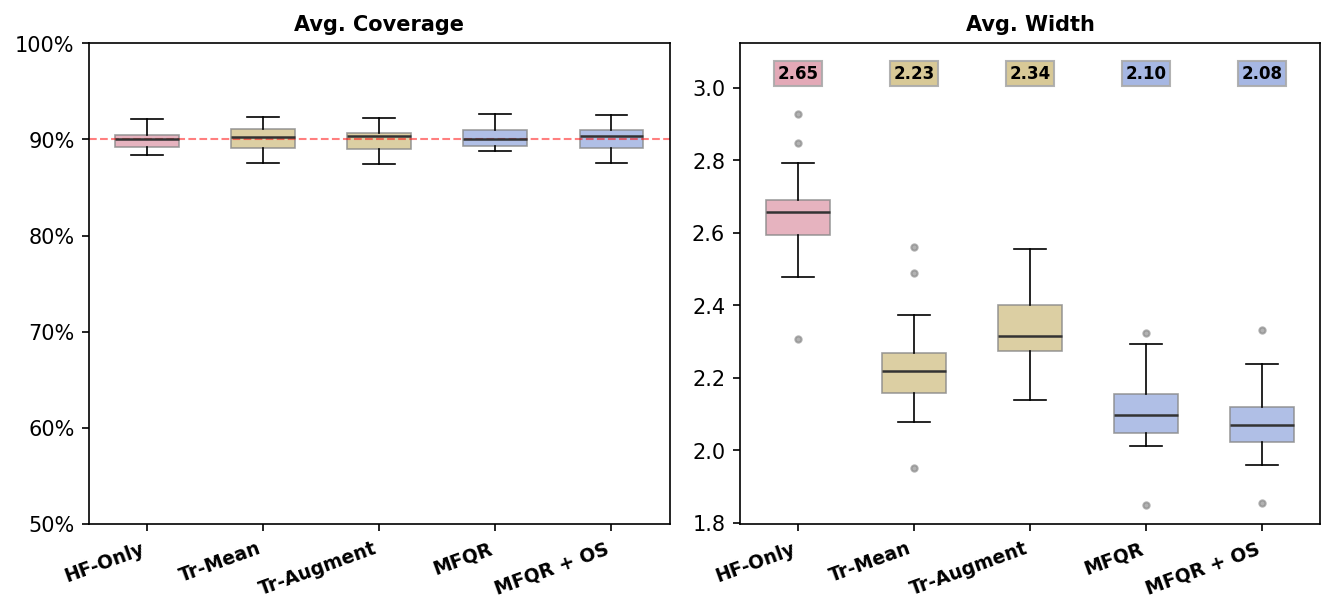}
        \caption{Thymine.}
    \end{subfigure}
    
    \vspace{0.5em}
    
    \begin{subfigure}[t]{0.9\textwidth}
        \centering
        \includegraphics[width=\textwidth]{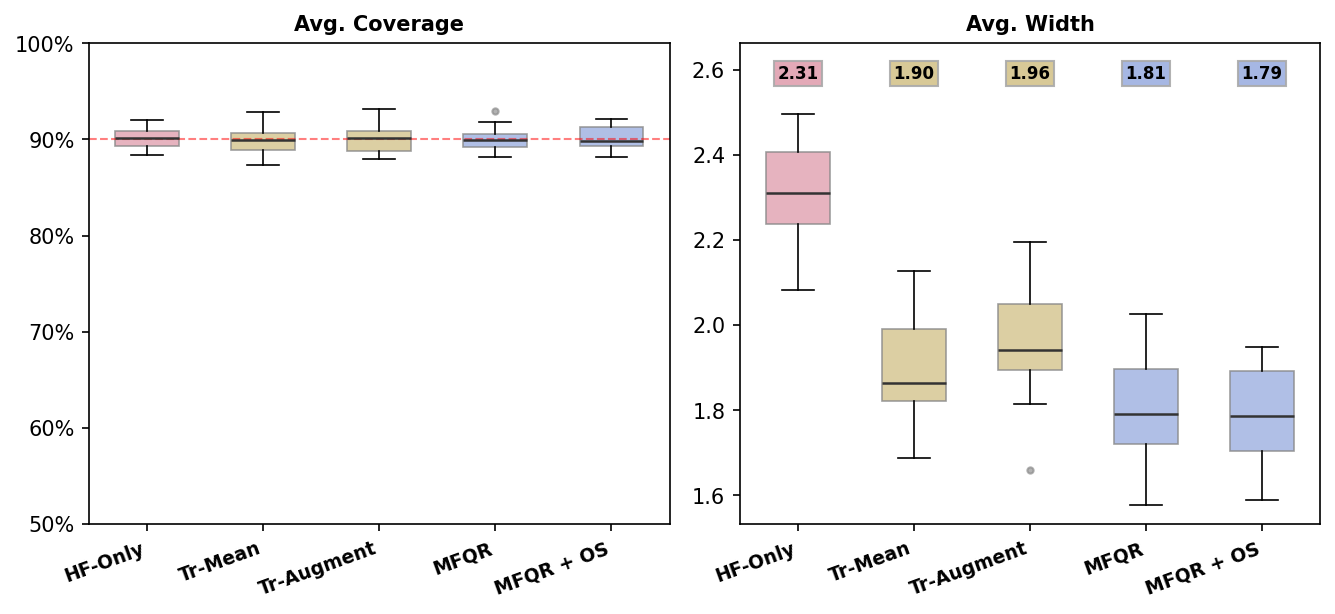}
        \caption{o-HBDI.}
    \end{subfigure}
\caption{Coverage and interval width on the QeMFi datasets.}
\label{fig:boxplots}
\end{figure}

\clearpage

\begin{figure}[t]
\vspace{-25pt}
    \centering
    \begin{subfigure}[t]{0.9\textwidth}
        \centering
        \includegraphics[width=\textwidth]{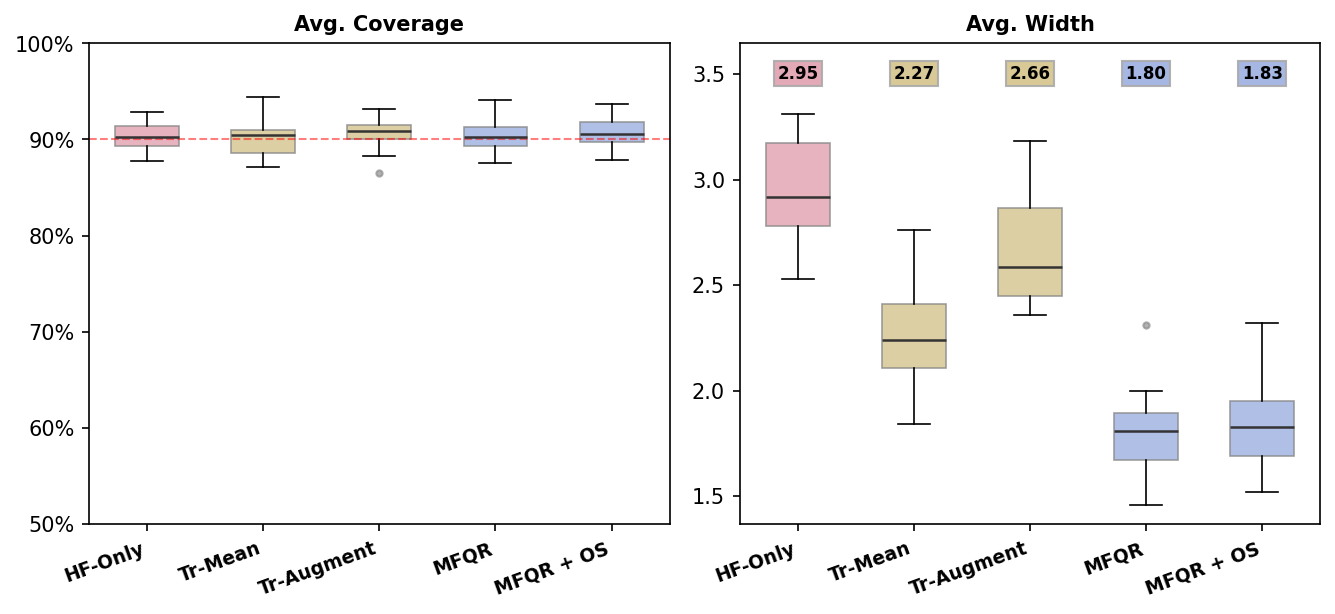}
        \caption{Burgers.}
    \end{subfigure}
    
    \vspace{0.5em}
    
    \begin{subfigure}[t]{0.9\textwidth}
        \centering
        \includegraphics[width=\textwidth]{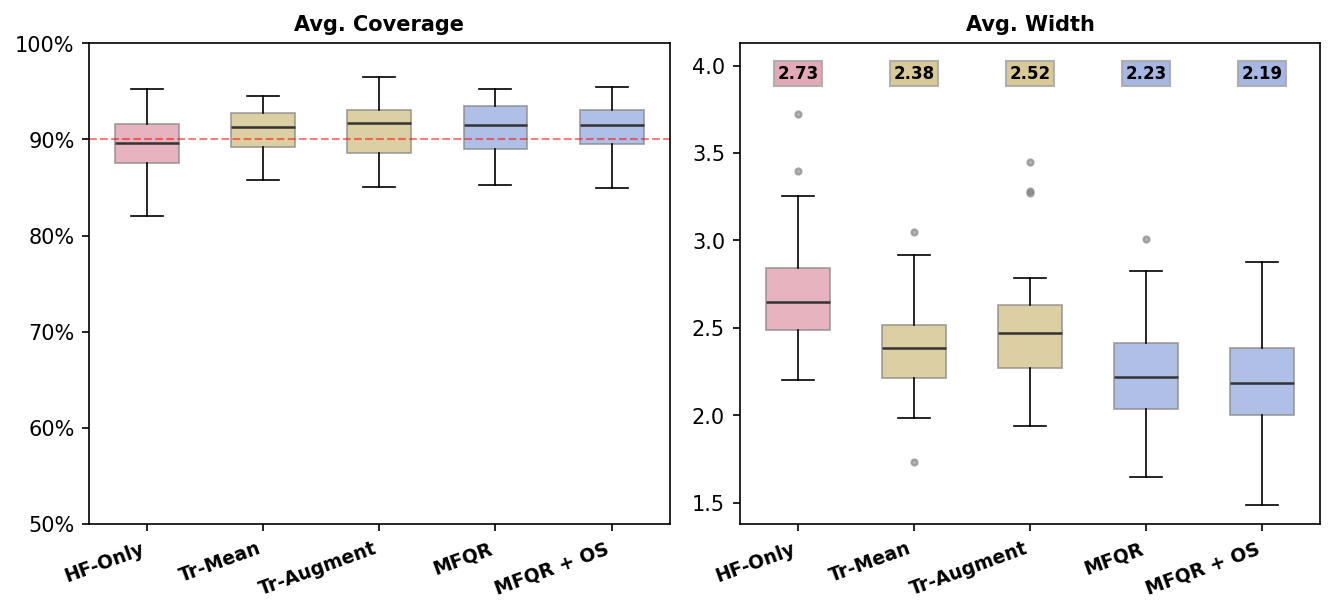}
        \caption{Formation Energy.}
    \end{subfigure}
\caption{Coverage and interval width on the Burgers and F-Energy datasets.}
\label{fig:boxplots_2}
\vspace{15pt}
\end{figure}

For each dataset, LF and HF responses are available on a common set of covariates. 
For Acrolein, Thymine, o-HBDI, and Formation Energy, we split the data into roughly \(50\%\) LF training, \(5\%\) HF training, \(5\%\) calibration, and \(40\%\) testing. 
For the low-dimensional Burgers dataset, we use a smaller HF training budget, \(n_H=60\), to keep the estimation problem more challenging.

Figures~\ref{fig:boxplots} and \ref{fig:boxplots_2} show the coverage and interval width over 20 random splits.
All methods achieve empirical coverage close to the nominal level \(90\%\), as expected from conformal prediction. 
Across the five datasets, MFQR and MFQR~+~OS produce the narrowest or nearly narrowest intervals, with width reductions of about \(20\)--\(39\%\) compared to HF-Only. 
The one-step correction gives a small additional improvement over MFQR on all datasets except Burgers, where the HF training size is deliberately small (\(n_H=60\)). 
In this case, the HF CDF estimator used in the correction step is too noisy to yield a meaningful improvement. 
This behavior is consistent with the expansion in Theorem~\ref{thm:multiple_bias_nonparam}: when the CDF estimation error \(\nu_\tau(x)\) is large relative to the wrapper error \(e_\tau(x)\), the correction can add variance rather than reduce error.

\section{Discussion}
\label{sec:discussion}

The proposed multi-fidelity quantile regression (MFQR) framework has two components that serve complementary purposes: transfer through the LF conditional distribution and targeted correction using the HF quantile equation. The wrapper estimator is most effective when the LF distribution captures the local structure of the HF distribution, because in that regime the level function \(u_\tau(x)\) can be easier to estimate than the original HF quantile surface. When the wrapper is less informative, the one-step correction uses the HF quantile equation to reduce bias. This complementary structure is reflected in our experiments, where the same framework performs well across both synthetic and scientific datasets in favorable and unfavorable regimes.

An important feature of MFQR is that it is model-agnostic: it only requires an estimator of the LF conditional distribution function, rather than a specific parametric form, learning architecture, or access to the original LF training data. This makes the framework compatible with pretrained LF models, including tabular foundation models \citep{hollmann2025accurate,qu2025tabicl}, whenever their outputs can be converted into an estimator of the LF conditional distribution function. This flexibility is especially important in scientific applications, where low-fidelity information may come from large pretrained models, simulations, or historical datasets, while the most reliable high-fidelity labels often come from small numbers of expensive laboratory experiments \citep{horawalavithana2022foundation,pyzer2025foundation}. 

A number of directions for future work remain open. One is to extend the framework beyond two fidelity levels by using the HF sample to select the most informative LF source, or more generally, to learn an HF-guided mixture wrapper that combines multiple LF distributions. Another is to study covariate shift, where LF and HF data may no longer be aligned in covariate space for reliable transfer. Such an extension would broaden the practical applicability of the method. Finally, it would be interesting to move beyond scalar outcomes. In our setting, the conditional distribution serves as a generative model of a univariate response. Extending the wrapper idea to structured outputs such as text and images would require a suitable multivariate analogue. We leave this direction for future work.

\bibliographystyle{plainnat}
\bibliography{references}

\clearpage

\appendix
\appendixpage

\section{Experiments with Gaussian process models}
\label{app:gp_results}

Section~\ref{sec:scientific} reports the scientific-data experiments using random forests. 
Here we repeat the same experiments using Gaussian process (GP) regression models, while keeping the data splits, random seeds, method definitions, and conformal prediction step unchanged. 
We use the same LF and HF GP models across all methods.
In the GP implementation, we approximate the infinite-dimensional feature map of the Gaussian kernel using \(1000\) random Fourier features \citep{rahimi2008random}. 
For each fitted GP model, we obtain a Gaussian predictive distribution
\[
Y(x)\mid x,\mathcal D
\sim
N\!\left(
\hat\mu_{\mathrm{GP}}(x),
\hat\sigma_{\mathrm{GP}}^2(x)
\right).
\]
We use this predictive distribution to approximate the conditional response distribution. 
The GP CDF and quantile estimates are defined as
\[
\hat F^{\mathrm{GP}}(y\mid x)
=
\Phi\!\left(
\frac{y-\hat\mu_{\mathrm{GP}}(x)}
{\hat\sigma_{\mathrm{GP}}(x)}
\right),
\qquad
\hat Q^{\mathrm{GP}}(\alpha\mid x)
=
\hat\mu_{\mathrm{GP}}(x)
+
\hat\sigma_{\mathrm{GP}}(x)\Phi^{-1}(\alpha).
\]
All other steps are unchanged from the random-forest implementation.

Figures~\ref{fig:gp_boxplots} and \ref{fig:gp_boxplots_2} show the coverage and interval width over 20 random splits. 
The results are broadly consistent with the random-forest experiments in Section~\ref{sec:scientific}. 
MFQR or MFQR+OS generally produces the narrowest prediction intervals. This suggests that the empirical advantage of MFQR is not tied to a particular conditional distribution estimator, but rather to the way it uses the LF distribution to define a smoother wrapper target.

The effect of the one-step correction varies across datasets. 
It improves over MFQR on Acrolein, Thymine, and o-HBDI. 
On Burgers, however, the correction gives wider intervals. 
As discussed above, the correction helps only when the HF CDF estimation error is sufficiently controlled. 
When the HF CDF estimate is noisy, the correction can introduce additional variation rather than reduce error. 
Overall, the GP experiments support the same conclusion as the random-forest experiments. 
MFQR can improve HF quantile estimation by leveraging LF data, while the correction step is most beneficial when the HF quantile equation can be estimated accurately enough.

\begin{figure}[H]
\vspace{-25pt}
    \centering
    \begin{subfigure}[t]{0.9\textwidth}
        \centering
        \includegraphics[width=\textwidth]{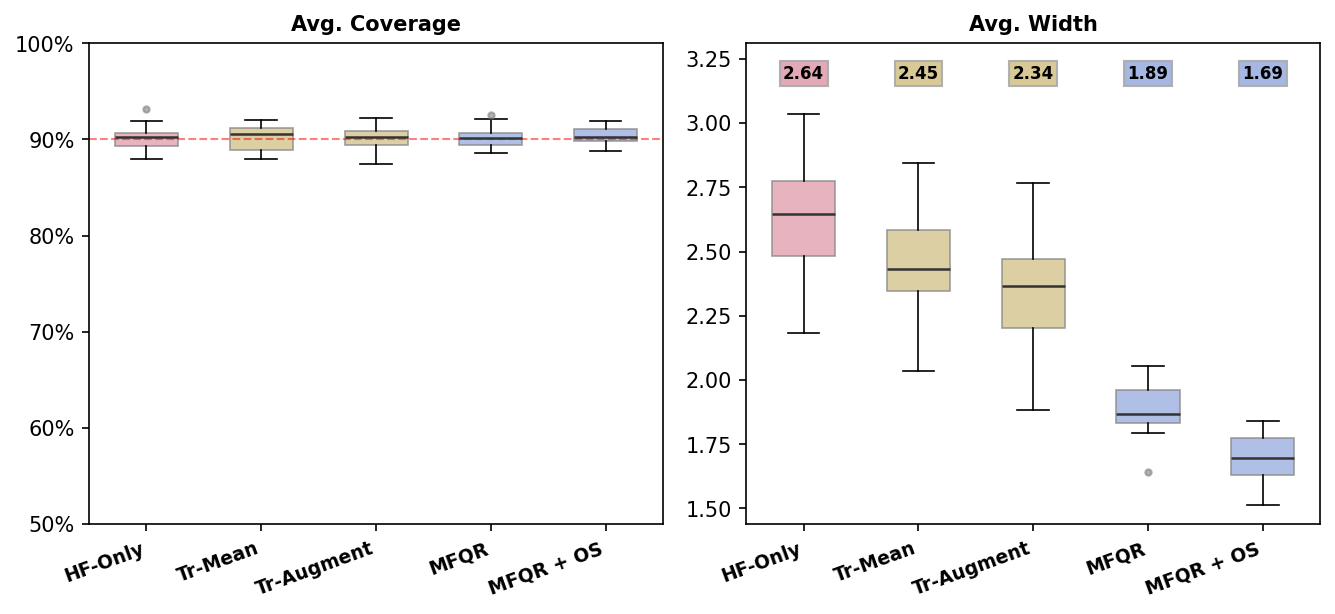}
        \caption{Acrolein.}
    \end{subfigure}
    \begin{subfigure}[t]{0.9\textwidth}
        \centering
        \includegraphics[width=\textwidth]{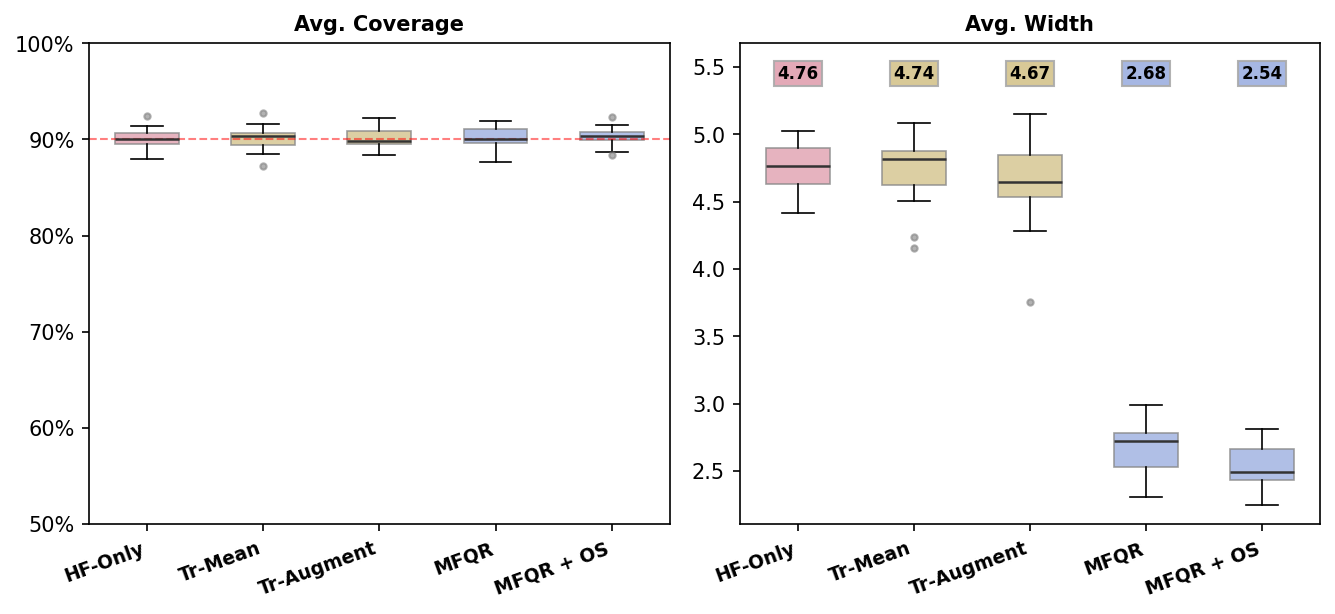}
        \caption{Thymine.}
    \end{subfigure}
    \begin{subfigure}[t]{0.9\textwidth}
        \centering
        \includegraphics[width=\textwidth]{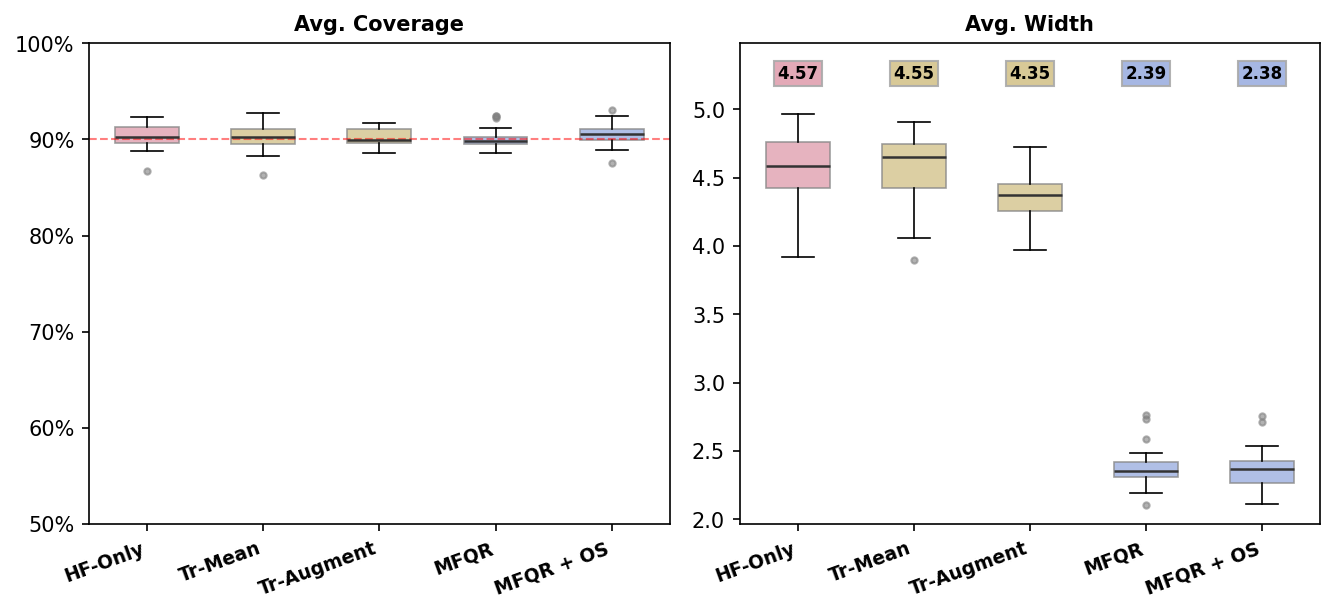}
        \caption{o-HBDI.}
    \end{subfigure}
\caption{Coverage and interval width on the QeMFi datasets under the Gaussian process implementation.}
\label{fig:gp_boxplots}
\end{figure}

\clearpage

\begin{figure}[H]
\vspace{-25pt}
    \centering
    \begin{subfigure}[t]{0.9\textwidth}
        \centering
        \includegraphics[width=\textwidth]{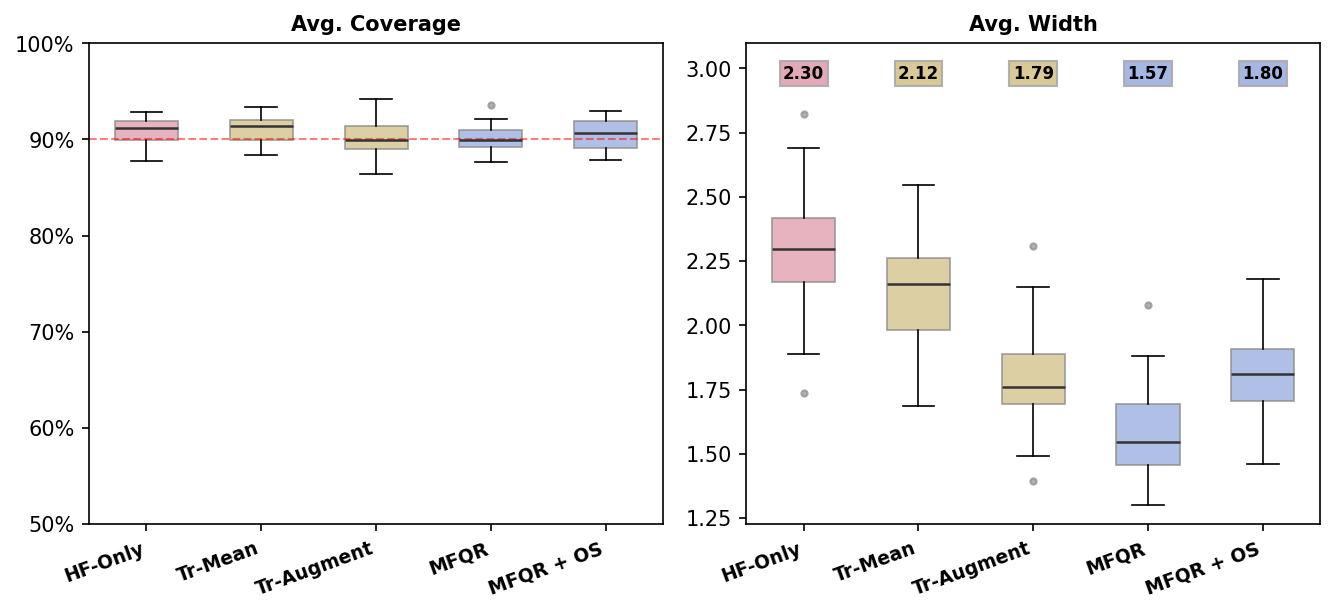}
        \caption{Burgers.}
    \end{subfigure}
    
    \vspace{0.5em}
    
    \begin{subfigure}[t]{0.9\textwidth}
        \centering
        \includegraphics[width=\textwidth]{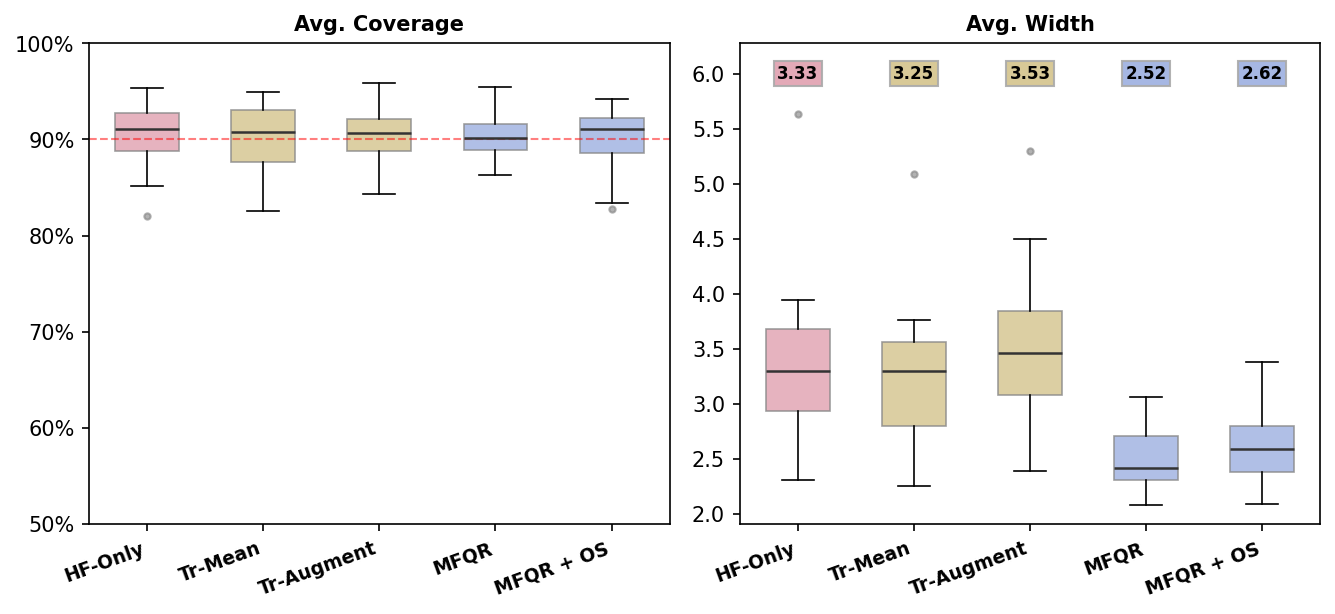}
        \caption{F-Energy.}
    \end{subfigure}
\caption{Coverage and interval width on the Burgers and F-Energy datasets under the Gaussian process implementation.}
\label{fig:gp_boxplots_2}
\vspace{15pt}
\end{figure}

\section{Implementations of MFQR}
\label{app:implementation}

The wrapper construction in MFQR is model-agnostic and can be combined with a broad range of estimators. This appendix discusses several approaches for estimating the LF conditional CDF \(F^L(\cdot\mid x)\) and the level function \(u_\tau(x)\).


\paragraph{Location-scale models.}
When domain knowledge suggests a simple model for the LF distribution, one may fit a parametric or semiparametric working model. For example, one may assume that
\[
Y^L=\mu_L(X)+\sigma_L(X)\varepsilon,
\]
where the error variable \(\varepsilon\) follows a reference distribution such as Gaussian or Student-\(t\). Using regression estimators for \(\mu_L\) and \(\sigma_L\) yields
\[
\hat F^L(y\mid x)
=
\hat G\!\left(\frac{y-\hat\mu_L(x)}{\hat\sigma_L(x)}\right),
\]
where \(\hat G\) denotes the estimator of the fitted error CDF.

\paragraph{Distribution regression.}
A direct approach is to estimate the binary regression function
\(
P(Y^L\le y\mid X=x)
\)
over a grid of thresholds \(y_1<\cdots<y_M\). Specifically, for each threshold \(y_m\), one fits a classifier for the binary outcome
\(
\one_{\{Y_j^L\le y_m\}}
\)
as a function of \(X_j^L\). This yields estimates
\(
\hat F^L(y_m\mid x)
\)
on the grid, which can then be interpolated in \(y\). Logistic regression, random forests, boosted trees, and neural networks can all be used in this step.

Another approach is to first estimate the LF conditional quantile function \(Q^L(\alpha\mid x)\) over a dense grid of probability levels
\[
0<\alpha_1<\cdots<\alpha_M<1,
\]
and then recover the CDF by numerical inversion:
\[
\hat F^L(y\mid x)
:=
\sup\{\alpha_m:\hat Q^L(\alpha_m\mid x)\le y\}.
\]
This approach is convenient when conditional quantile estimators are easier to train or already available. It also aligns naturally with the final MFQR estimator, since the wrapper maps back through an LF quantile function.

\paragraph{Random forests.} \citet{meinshausen2006quantile} proposes a random-forest-based method to estimate the conditional distribution and quantiles. For a test point \(x\), each tree in the forest places \(x\) into a terminal leaf, and the training observations in the same leaf are treated as neighbors of \(x\). This induces adaptive weights on the training sample. If \(A_t(x)\) denotes the terminal leaf containing \(x\) in tree \(t\), then the tree-level weight on observation $i$ is
\[
w_i^{(t)}(x)
=
\frac{\one_{\{X_i\in A_t(x)\}}}{|A_t(x)|},
\]
and the forest weight is the average over trees,
\[
w_i(x)=\frac{1}{T}\sum_{t=1}^T w_i^{(t)}(x).
\]
These weights are nonnegative and sum to one, so they define a local empirical distribution around \(x\). Using these weights, the LF conditional CDF can be estimated by the weighted empirical CDF
\[
\hat F^L_{\mathrm{QRF}}(y\mid x)
=
\sum_{j=1}^{n_L} w_j(x)\,\one_{\{Y_j^L\le y\}}.
\]
The corresponding conditional quantile estimator is obtained by inversion,
\[
\hat Q^L_{\mathrm{QRF}}(\alpha\mid x)
=
\inf\{y:\hat F^L_{\mathrm{QRF}}(y\mid x)\ge \alpha\}.
\]
Thus, quantile regression forests are especially convenient in our setting because they provide both a conditional CDF estimator and a conditional quantile estimator through the same forest weights. 
When the correction step requires a conditional density estimate, we use the same forest weights and apply Gaussian-kernel smoothing to the weighted observations. To avoid numerical instability, we replace density estimates smaller than \(0.1\) by \(0.1\), and cap the magnitude of each correction step at \(1.0\).

The same idea can also be used in the second stage to estimate the wrapped target \(u_\tau(x)\). After forming the responses $\hat U_i=\hat F^L(Y_i^H\mid X_i^H)$,
one may fit a quantile regression forest to the pairs \((X_i^H,\hat U_i)\) and estimate
\[
\hat u_\tau(x)
=
\inf\Bigl\{u:\sum_{i=1}^{n_H}\tilde w_i(x)\,\one_{\{\hat U_i\le u\}}\ge \tau\Bigr\},
\]
where \(\tilde w_i(x)\) are the forest weights from the HF-stage forest. Since \(\hat U_i\in[0,1]\), this gives a flexible nonparametric estimator of the level function while preserving its interpretation as a conditional quantile on the probability scale.

\paragraph{Neural-network quantile regression \citep{taylor2000quantile}.}
Neural networks provide a flexible way to estimate \(u_\tau(x)\) when its dependence on \(x\) may be highly nonlinear. Let \(m_\theta(x)\) denote a neural-network predictor with parameters \(\theta\). One may estimate \(u_\tau(x)\) by minimizing the pinball loss
\[
\hat\theta_\tau
\in
\argmin_{\theta}
\sum_{i=1}^{n_H}
\rho_\tau\!\bigl(\hat U_i-m_\theta(X_i^H)\bigr),
\]
where
\(
\rho_\tau(z)=z\bigl(\tau-\one_{\{z<0\}}\bigr)
\). The quantile estimator is $\hat u_\tau(x)=m_{\hat\theta_\tau}(x).$  In practice, regularization is important to prevent overfitting, especially because the HF sample is typically limited. Standard approaches include weight decay, early stopping, dropout, and cross-validation for tuning model complexity. Neural networks can also be used in the LF stage, either by modeling the LF conditional CDF directly through distribution regression or by estimating a dense grid of LF conditional quantiles and recovering \(\hat F^L(y\mid x)\) by inversion.

\section{Technical proofs}

\subsection{Hölder classes}\label{section:holder}
Let \(\mathcal X\subset\mathbb R^p\) be compact, and let \(\beta>0\). Write
\[
m:=\lceil \beta\rceil-1,\qquad \alpha:=\beta-m\in(0,1].
\]
For a multi-index \(k=(k_1,\dots,k_p)\in\mathbb N_0^p\), let \(|k|:=k_1+\cdots+k_p\) and
\[
D^k f:=\frac{\partial^{|k|}f}{\partial x_1^{k_1}\cdots \partial x_p^{k_p}},
\]
with \(D^0 f=f\).
We say that \(f\) belongs to the Hölder class \(\mathcal H(\beta,L)\) if:
\begin{enumerate}
\item all mixed partial derivatives \(D^k f\) exist and are bounded on \(\mathcal X\) for every \(|k|\le m\);
\item for every multi-index \(k\) with \(|k|=m\),
\[
|D^k f(x)-D^k f(x')|
\le
L\,\|x-x'\|^\alpha,
\qquad x,x'\in\mathcal X.
\]
\end{enumerate}
When \(0<\beta\le 1\), this reduces to the usual Hölder condition
\[
|f(x)-f(x')|\le L\,\|x-x'\|^\beta,\qquad x,x'\in\mathcal X.
\]

\subsection{Proof of Proposition~\ref{prop:wrapper_smoother_extended}}

\begin{proof}
Since
\[
q_\tau(x)=\mu(x)+\sigma(x)G_H^{-1}(\tau),
\]
it is the sum of \(\mu(x)\) and a constant multiple of \(\sigma(x)\). Therefore its Hölder smoothness is determined by the rougher of these two components, so
\[
q_\tau \in \mathcal H(\beta_q,C_q),
\qquad
\beta_q=\min\{\beta_\mu,\beta_\sigma\},
\]
for some finite constant \(C_q>0\) depending only on \(\tau\), \(G_H\), \(C_\mu\), and \(C_\sigma\).

Next,
\[
r_\tau(x)=q_\tau(x)-\mu(x)=\sigma(x)G_H^{-1}(\tau),
\]
so \(r_\tau\) is a constant multiple of \(\sigma\). By Assumption~\ref{ass:holder_mu},
\[
r_\tau \in \mathcal H\!\left(\beta_\sigma,\; |G_H^{-1}(\tau)|C_\sigma\right).
\]

Finally, by definition
\[
u_\tau(x):=F^L(q_\tau(x)\mid x),
\qquad
q_\tau(x)=\mu(x)+\sigma(x)G_H^{-1}(\tau).
\]
Under model \eqref{eq:two_model},
\[
F^L(y\mid x)=G_L\!\left(\frac{y-\mu(x)}{\sigma(x)\rho(x)}\right).
\]
Evaluating at \(y=q_\tau(x)\) gives
\[
u_\tau(x)
=
G_L\!\left(\frac{G_H^{-1}(\tau)}{\rho(x)}\right).
\]
Define
\[
\psi_\tau(r):=G_L\!\left(r^{-1}G_H^{-1}(\tau)\right),\qquad r>0.
\]
Then
\[
u_\tau(x)=\psi_\tau(\rho(x)).
\]

Since \(\rho(x)\) is bounded away from zero and \(G_L\in C^{\lceil\beta_\rho\rceil}(\mathbb R)\) by Assumption~\ref{ass:G_regular}, the map \(\psi_\tau\) is \(C^{\lceil\beta_\rho\rceil}\) on the range of \(\rho\). Combining this with Assumption~\ref{ass:holder_mu} and the standard composition property of Hölder classes yields
\[
u_\tau \in \mathcal H\!\left(\beta_\rho,\; C_\tau C_\rho\right)
\]
for some finite constant \(C_\tau>0\).
\end{proof}

\subsection{Proof of Theorem~\ref{thm:wrapper_nonparam}}

\begin{proof}
Fix \(x\in\mathcal X\) and \(\tau\in(0,1)\). Let
\[
Q^L(u\mid x):=(F^L)^{-1}(u\mid x),
\qquad
\hat Q^L(u\mid x):=(\hat F^L)^{-1}(u\mid x).
\]
By Assumption~\ref{assume:qq}, \(q_\tau(x)=Q^L(u_\tau(x)\mid x)\),
and by definition of the wrapper estimator,
\(\tilde q_\tau(x)=\hat Q^L(\hat u_\tau(x)\mid x)\).
Hence
\begin{align*}
\tilde q_\tau(x)-q_\tau(x)
= &\ 
\Bigl[\hat Q^L(\hat u_\tau(x)\mid x)-Q^L(\hat u_\tau(x)\mid x)\Bigr] \\
& +
\Bigl[Q^L(\hat u_\tau(x)\mid x)-Q^L(u_\tau(x)\mid x)\Bigr].
\end{align*}
We control the two brackets separately.

\medskip
\noindent
\emph{Step 1: control of \(Q^L(\hat u_\tau(x)\mid x)-Q^L(u_\tau(x)\mid x)\).} 
Let \(\mathcal Y\) denote the compact interval from Assumption~\ref{ass:global_locality} on which \(F^L(\cdot\mid x)\) is strictly increasing and \(f^L(\cdot\mid x)\ge c_L>0\). Since \(q_\tau(x)\in \operatorname{int}(\mathcal Y)\) and
\[
u_\tau(x)=F^L(q_\tau(x)\mid x),
\]
the point \(u_\tau(x)\) lies in the interior of
\[
J_x:=F^L(\mathcal Y\mid x).
\]
By Assumption~\ref{ass:global_locality}, the map \(Q^L(\cdot\mid x)\) is twice continuously differentiable on a neighborhood of \(u_\tau(x)\). Hence there exists an open interval \(U_x\subset J_x\) containing \(u_\tau(x)\) on which \(\partial_u Q^L(u\mid x)\) is bounded. Since
\[
\hat u_\tau(x)-u_\tau(x)=O_p(a_n)
\qquad\text{and}\qquad
a_n\to0,
\]
we have \(\hat u_\tau(x)\in U_x\) with probability tending to one. A first-order Taylor expansion then gives
\[
Q^L(\hat u_\tau(x)\mid x)-Q^L(u_\tau(x)\mid x)
=
\partial_u Q^L(\bar u_\tau(x)\mid x)\bigl(\hat u_\tau(x)-u_\tau(x)\bigr),
\]
for some random \(\bar u_\tau(x)\) between \(\hat u_\tau(x)\) and \(u_\tau(x)\). Since \(\partial_u Q^L(\cdot\mid x)\) is bounded on \(U_x\),
\[
Q^L(\hat u_\tau(x)\mid x)-Q^L(u_\tau(x)\mid x)
=
O_p(a_n).
\]

\medskip
\noindent
\emph{Step 2: control of  \(\hat Q^L(\hat u_\tau(x)\mid x)-Q^L(\hat u_\tau(x)\mid x)\).}
Define
\[
\varepsilon_n(x):=\sup_{y\in \mathcal Y}\bigl|\hat F^L(y\mid x)-F^L(y\mid x)\bigr|.
\]
By assumption, $\varepsilon_n(x)=O_p(b_n),~b_n\to0$.
Also, by assumption, \(\hat F^L(\cdot\mid x)\) is nondecreasing in \(y\) with probability tending to one.
Because \(u_\tau(x)\in \operatorname{int}(J_x)\) and \(\hat u_\tau(x)-u_\tau(x)=O_p(a_n)\) with \(a_n\to0\), there exists a smaller closed interval \(J_x'\subset \operatorname{int}(J_x)\) containing \(u_\tau(x)\) such that
\[
\hat u_\tau(x)\in J_x'
\]
with probability tending to one. Since \(J_x'\subset \operatorname{int}(J_x)\) is compact, there exists
\(d_x>0\) such that every point in \(J_x'\) has distance at least \(d_x\)
from the boundary of \(J_x\). Because \(\varepsilon_n(x)\to0\) in probability,
we have \(\varepsilon_n(x)<d_x\) with probability tending to one. Hence, on an event
with probability tending to one,
\[
\hat u_\tau(x)\pm \varepsilon_n(x)\in J_x.
\]
Fix \(u\in J_x\). On the event that \(\hat F^L(\cdot\mid x)\) is nondecreasing and
\[
\sup_{y\in\mathcal Y}\bigl|\hat F^L(y\mid x)-F^L(y\mid x)\bigr|\le \varepsilon_n(x),
\]
monotonicity of generalized inverses implies
\[
Q^L(u-\varepsilon_n(x)\mid x)\le \hat Q^L(u\mid x)\le Q^L(u+\varepsilon_n(x)\mid x).
\]
Applying this with \(u=\hat u_\tau(x)\) yields
\[
Q^L(\hat u_\tau(x)-\varepsilon_n(x)\mid x)
\le
\hat Q^L(\hat u_\tau(x)\mid x)
\le
Q^L(\hat u_\tau(x)+\varepsilon_n(x)\mid x).
\]

Since \(f^L(y\mid x)\ge c_L>0\) on \(\mathcal Y\), the inverse map \(Q^L(\cdot\mid x)\) is Lipschitz on \(J_x\), with constant at most \(1/c_L\). Therefore
\[
\bigl|\hat Q^L(\hat u_\tau(x)\mid x)-Q^L(\hat u_\tau(x)\mid x)\bigr|
\le
\frac{1}{c_L}\,\varepsilon_n(x)
=
O_p(b_n).
\]

\medskip
\noindent
Combining the two steps gives
\[
\tilde q_\tau(x)-q_\tau(x)
=
O_p(a_n+b_n),
\]
which proves the theorem.
\end{proof}

\subsection{Local and uniform kernel estimator bounds}
\label{app:kernel_quantile_rates}

This subsection collects the kernel estimator bounds used to prove
\Cref{cor:wrapper_rate_transfer}. We first give pointwise bounds at a fixed
covariate value \(x\). These bounds control the local conditional CDF estimator
and the corresponding local quantile estimator at that point.
We then prove a stronger uniform bound for the conditional CDF estimator. More precisely, this bound controls
\[
\sup_{x\in\mathcal X,\;z\in\mathcal Y}
\left|\hat F^L(z\mid x)-F^L(z\mid x)\right|,
\]
where \(z\) denotes the response value at which the CDF is evaluated. This
uniform bound is needed because the wrapper constructs generated pseudo-responses
$\hat U_i=\hat F^L(Y_i^H\mid X_i^H)$,
so the first-stage LF CDF error must be controlled over all relevant covariate
and response values. This error then propagates into the estimation of the level
function \(u_\tau(x)\).

Compared with the fixed-\(x\) bound, the uniform bound has an extra logarithmic
factor because it controls the estimator over the whole covariate space
\(\mathcal X\), rather than at one fixed point. The reason can be briefly
explained as follows. In the proof, we approximate \(\mathcal X\) by a finite
set of representative points
$\mathcal N_n=\{x_1,\dots,x_{N_n}\}$,
chosen so that every \(x\in\mathcal X\) is within distance \(\delta_n\) of some
grid point \(x_j\). At each grid point, the random error has the usual order
$(nh^p)^{-1/2}$.
Controlling all \(N_n\) grid points at once adds a factor of order
$\sqrt{\log N_n}$.
Since the grid is chosen so that \(\log N_n=O(\log n)\), the random error in the
uniform bound becomes
$\sqrt{\frac{\log n}{nh^p}}$.
The Lipschitz continuity of the kernel then extends the bound from the grid
points to all \(x\in\mathcal X\).

Throughout this subsection, let \(\{(X_i,Z_i)\}_{i=1}^n\) be i.i.d. copies of
\((X,Z)\).

\begin{assumption}
\label{ass:kernel_aux_setup}
Assume that \(\mathcal X\subset\mathbb R^p\) is compact. Let \(Z\) be a scalar response with conditional CDF \(F_Z(\cdot\mid x)\),
conditional density \(f_{Z\mid X}(\cdot\mid x)\), and conditional
\(\tau\)-quantile \(m_\tau(x)\). Fix \(\tau\in(0,1)\). Assume there exists a
compact interval \(\mathcal Z\subset\mathbb R\) such that the following
conditions hold for all \(x\in\mathcal X\):
\begin{enumerate}[label=(\roman*)]
    \item \(m_\tau(x)\) lies in the interior of \(\mathcal Z\), and the conditional support of \(Z\mid X=x\) is contained in \(\mathcal Z\);
    \item \(F_Z(\cdot\mid x)\in \mathcal C^1(\mathcal Z)\), with
    \[
    f_{Z\mid X}(z\mid x)\ge c_Z>0
    \qquad \text{for all } z\in \mathcal Z;
    \]
    \item for every \(z\in \mathcal Z\), the map \(x'\mapsto F_Z(z\mid x')\) belongs to \(\mathcal H(\beta,C_Z)\), where \(\beta\in(0,1]\);
    \item the covariate distribution has a density \(f_X\) satisfying
    \[
    0<c_X\le f_X(x')\le C_X<\infty
    \qquad \text{for all }x'\in\mathcal X,
    \]
    and, for the kernel \(K\) in part (v), there exist constants \(h_0>0\) and \(c_{\mathcal X}>0\) such that
    \[
    \inf_{0<h<h_0}\inf_{x'\in\mathcal X}
    \int_{\mathcal X} h^{-p}K((t-x')/h)\,dt
    \ge c_{\mathcal X};
    \]
    \item the kernel \(K:\mathbb R^p\to\mathbb R\) is bounded, Lipschitz, compactly supported, nonnegative, and satisfies
    \[
    \int K(u)\,du=1,
    \qquad
    \int K(u)^2\,du<\infty;
    \]
    \item the bandwidth satisfies \(h\to0\) and \(nh^p\to\infty\).
\end{enumerate}
\end{assumption}

For a bandwidth \(h>0\), let
\[
K_h(z):=h^{-p}K(z/h),
\qquad
w_i(x;h):=\frac{K_h(X_i-x)}{\sum_{j=1}^n K_h(X_j-x)}.
\]

Define the kernel denominator
\[
\hat f_h(x):=\frac{1}{n}\sum_{i=1}^n K_h(X_i-x),
\qquad
m_h(x):=\E\{\hat f_h(x)\}.
\]

\begin{lemma}[Kernel denominator bounds]
\label{lem:kernel_denominator_uniform}
Under Assumption~\ref{ass:kernel_aux_setup}, for each fixed \(x\in\mathcal X\),
\[
\hat f_h(x)-m_h(x)=O_p\!\left((nh^p)^{-1/2}\right).
\]
If additionally \(nh^p/\log n\to\infty\), then
\[
\sup_{x\in\mathcal X}
\left|
\hat f_h(x)-m_h(x)
\right|
=
O_p\!\left(\sqrt{\frac{\log n}{nh^p}}\right).
\]
\end{lemma}

For later use, define
\[
G_n(x,z)
:=
\frac{1}{n}\sum_{i=1}^n
K_h(X_i-x)
\{\one\{Z_i\le z\}-F_Z(z\mid X_i)\}.
\]

\begin{lemma}[Centered kernel numerator bounds]
\label{lem:kernel_emp_process}
Under Assumption~\ref{ass:kernel_aux_setup}, for each fixed \(x\in\mathcal X\),
\[
\sup_{z\in\mathcal Z}|G_n(x,z)|
=
O_p\!\left((nh^p)^{-1/2}\right).
\]
If additionally \(nh^p/\log n\to\infty\), then
\[
\sup_{x\in\mathcal X,\;z\in\mathcal Z}|G_n(x,z)|
=
O_p\!\left(\sqrt{\frac{\log n}{nh^p}}\right).
\]
\end{lemma}

\begin{proposition}[Pointwise CDF rate]
\label{prop:kernel_cdf_local_uniform_aux}
Under Assumption~\ref{ass:kernel_aux_setup}, for each fixed \(x\in\mathcal X\),
\[
\sup_{z\in \mathcal Z}\bigl|\hat F_Z(z\mid x)-F_Z(z\mid x)\bigr|
=
O_p\!\bigl(h^\beta+(n h^p)^{-1/2}\bigr),
\]
where
\[
\hat F_Z(z\mid x):=\sum_{i=1}^n w_i(x;h)\,\one\{Z_i\le z\}.
\]
\end{proposition}

\begin{proposition}[Pointwise quantile rate]
\label{prop:kernel_quantile_rate_aux}
Under Assumption~\ref{ass:kernel_aux_setup}, for each fixed \(x\in\mathcal X\), let
\[
\hat m_\tau(x):=\inf\{z:\hat F_Z(z\mid x)\ge \tau\},
\]
where \(\hat F_Z(\cdot\mid x)\) is the kernel conditional CDF estimator in Proposition~\ref{prop:kernel_cdf_local_uniform_aux}. Then
\[
\hat m_\tau(x)-m_\tau(x)
=
O_p\!\bigl(h^\beta+(n h^p)^{-1/2}\bigr).
\]
Consequently, with the MSE-optimal bandwidth
\[
h^\star \asymp n^{-1/(2\beta+p)},
\]
we have
\[
\hat m_\tau(x)-m_\tau(x)
=
O_p\!\left(n^{-\beta/(2\beta+p)}\right).
\]
\end{proposition}

\begin{proposition}[Uniform CDF rate]
\label{prop:kernel_cdf_uniform_aux}
Under Assumption~\ref{ass:kernel_aux_setup}, if additionally \(nh^p/\log n\to\infty\), then
\[
\sup_{x\in\mathcal X,\; z\in \mathcal Z}
\bigl|\hat F_Z(z\mid x)-F_Z(z\mid x)\bigr|
=
O_p\!\left(h^\beta+\sqrt{\frac{\log n}{n h^p}}\right).
\]
\end{proposition}

\subsubsection{Proof of \Cref{lem:kernel_denominator_uniform}}

\begin{proof}
For a fixed \(x\), 
\[
\hat f_h(x)-m_h(x)
=
\frac{1}{n}\sum_{i=1}^n
\{K_h(X_i-x)-\E K_h(X_i-x)\}.
\]
Since \(K\) is bounded and compactly supported,
\[
\E K_h(X_i-x)^2 \le C h^{-p}.
\]
Thus
\[
\Var\{\hat f_h(x)\}=O((nh^p)^{-1}),
\]
which gives the pointwise bound by Chebyshev's inequality.

For the uniform bound, let
\[
r_n:=\sqrt{\frac{\log n}{nh^p}},
\qquad
\delta_n:=h^{p+1}r_n.
\]
Cover \(\mathcal X\) by a finite grid
\(\mathcal N_n=\{x_1,\dots,x_{N_n}\}\) with mesh width \(\delta_n\). 
The covering number satisfies \(N_n\le C\delta_n^{-p}\). Since
\(nh^p/\log n\to\infty\), for all sufficiently large \(n\),
\(h^p\ge \log n/n\), and hence \(\log(1/h)=O(\log n)\). Also, because
\(h\to0\), we have \(h^p\le 1\) for all sufficiently large \(n\), so
\[
\sqrt{\frac{\log n}{nh^p}}
\ge
\sqrt{\frac{\log n}{n}},
\]
which implies \(\log(1/r_n)=O(\log n)\). Therefore
\[
\log N_n
\le C+p\log(1/\delta_n)
=
O(\log n).
\]
For each grid point \(x_j\), define
\[
Y_i(x_j):=K_h(X_i-x_j)-\E K_h(X_i-x_j).
\]
Then \(\E [Y_i(x_j)]=0\),
\[
|Y_i(x_j)|\le C h^{-p},
\qquad
\Var\{Y_i(x_j)\}\le C h^{-p}.
\]
By Bernstein's inequality, for any \(t>0\),
\[
\Prob\left(
\left|\hat f_h(x_j)-m_h(x_j)\right|>t
\right)
\le
2\exp\left[
-\frac{c n t^2}{h^{-p}+h^{-p}t}
\right].
\]
Taking $t=C_1\sqrt{ (\log n) / (n h^p)}$,
and using \(nh^p/\log n\to\infty\), the second Bernstein term is negligible.
Thus, for \(C_1\) sufficiently large,
\[
\Prob\left(
\left|\hat f_h(x_j)-m_h(x_j)\right|>C_1\sqrt{\frac{\log n}{n h^p}}
\right)
\le
2\exp(-C_2\log n).
\]
Taking a union bound over \(j=1,\dots,N_n\), and choosing \(C_1\) large enough
so that \(C_2\log n-\log N_n\to\infty\), gives
\[
\max_{1\le j\le N_n}
\left|\hat f_h(x_j)-m_h(x_j)\right|
=
O_p\!\left(\sqrt{\frac{\log n}{nh^p}}\right).
\]
Since \(K\) is Lipschitz,
\[
|\hat f_h(x)-\hat f_h(x_j)|
\le
C h^{-p-1}\|x-x_j\|,
\]
and, by taking expectations,
\[
|m_h(x)-m_h(x_j)|
\le
C h^{-p-1}\|x-x_j\|.
\]
Thus both terms vary by at most \(C h^{-p-1}\delta_n\) between \(x\) and its nearest grid point.
By the choice \(\delta_n=h^{p+1}r_n\), this interpolation error is bounded by
\(C h^{-p-1}\delta_n=C r_n\), where
\(r_n=\sqrt{\log n/(nh^p)}\).  Therefore
\[
\sup_{x\in\mathcal X}
|\hat f_h(x)-m_h(x)|
=
O_p\!\left(\sqrt{\frac{\log n}{nh^p}}\right).
\]
\end{proof}

\subsubsection{Proof of \Cref{lem:kernel_emp_process}}

\begin{proof}
We prove the two claims separately.

\medskip
\noindent
Fix \(x\in\mathcal X\), and condition on \(X_1,\dots,X_n\). Write
\[
a_i(x):=\frac{1}{n}K_h(X_i-x),
\qquad
A_n(x)^2:=\sum_{i=1}^n a_i(x)^2.
\]
For each \(z\in\mathcal Z\), define
\[
\xi_i(z):=\one\{Z_i\le z\}-F_Z(z\mid X_i).
\]
Conditional on \(X_1,\dots,X_n\), the variables \(\xi_i(z)\) are independent and mean zero for every fixed \(z\).

We first bound the conditional expectation of
\[
S_n(x):=\sup_{z\in\mathcal Z}\left|\sum_{i=1}^n a_i(x)\xi_i(z)\right|.
\]
Let \(Z_1',\dots,Z_n'\) be an independent copy of \(Z_1,\dots,Z_n\) conditional on \(X_1,\dots,X_n\), and let
\(\varepsilon_1,\dots,\varepsilon_n\) be independent Rademacher random variables, independent of all other variables. By symmetrization,
\begin{align*}
\E\{S_n(x)\mid X_1,\dots,X_n\}
&\le
\E\left[
\sup_{z\in\mathcal Z}
\left|
\sum_{i=1}^n a_i(x)
\{\one\{Z_i\le z\}-\one\{Z_i'\le z\}\}
\right|
\,\middle|\, X_1,\dots,X_n
\right] \\
&\le
2\E\left[
\sup_{z\in\mathcal Z}
\left|
\sum_{i=1}^n a_i(x)\varepsilon_i\one\{Z_i\le z\}
\right|
\,\middle|\, X_1,\dots,X_n
\right].
\end{align*}
Now condition further on \(Z_1,\dots,Z_n\). Order the observations so that
\[
Z_{(1)}\le Z_{(2)}\le \cdots \le Z_{(n)},
\]
and let \(a_{(i)}(x)\) be the corresponding reordered weights. As \(z\) varies, the sum
\[
\sum_{i=1}^n a_i(x)\varepsilon_i\one\{Z_i\le z\}
\]
takes values among the partial sums
\[
M_m:=\sum_{i=1}^m a_{(i)}(x)\varepsilon_{(i)},
\qquad m=0,1,\dots,n.
\]
The process \((M_m)_{m=0}^n\) is a martingale with respect to the filtration generated by
\(\varepsilon_{(1)},\dots,\varepsilon_{(m)}\). By Doob's \(L^2\) maximal inequality,
\[
\E_\varepsilon\left[
\max_{0\le m\le n}|M_m|^2
\,\middle|\, Z_1,\dots,Z_n,X_1,\dots,X_n
\right]
\le
4\,\E_\varepsilon\left[
|M_n|^2
\,\middle|\, Z_1,\dots,Z_n,X_1,\dots,X_n
\right].
\]
Since the Rademacher variables are independent and mean zero,
\[
\E_\varepsilon |M_n|^2
=
\sum_{i=1}^n a_i(x)^2
=
A_n(x)^2.
\]
Therefore,
\[
\E_\varepsilon\left[
\max_{0\le m\le n}|M_m|
\,\middle|\, Z_1,\dots,Z_n,X_1,\dots,X_n
\right]
\le
2A_n(x).
\]
Combining the previous displays gives
\[
\E\{S_n(x)\mid X_1,\dots,X_n\}
\le
4A_n(x).
\]

It remains to control \(A_n(x)\). Since \(K\) is bounded and compactly supported, and since \(f_X\le C_X\),
\[
\E\{K_h(X_i-x)^2\}
=
\int_{\mathcal X} h^{-2p}K((t-x)/h)^2 f_X(t)\,dt
\le
C h^{-p}
\]
for a finite constant \(C\). Hence
\[
\E\{A_n(x)^2\}
=
\frac{1}{n^2}\sum_{i=1}^n
\E\{K_h(X_i-x)^2\}
\le
\frac{C}{nh^p}.
\]
Moreover,
\[
\E A_n(x)
\le
\{\E A_n(x)^2\}^{1/2}
\le
C (nh^p)^{-1/2}.
\]
Taking expectations in
\(\E\{S_n(x)\mid X_1,\dots,X_n\}\le 4A_n(x)\) gives
\[
\E S_n(x)\le C (nh^p)^{-1/2}.
\]
Markov's inequality then yields
\[
S_n(x)=O_p\!\left((nh^p)^{-1/2}\right),
\]
which proves the fixed-\(x\) claim.

\medskip
\noindent
Let
\[
r_n:=\sqrt{\frac{\log n}{nh^p}},
\qquad
\delta_n:=h^{p+1}r_n.
\]
Because \(nh^p/\log n\to\infty\), we have \(r_n\to0\). Since \(\mathcal X\) is compact, there exists a finite grid
\[
\mathcal N_n=\{x_1,\dots,x_{N_n}\}\subset\mathcal X
\]
such that every \(x\in\mathcal X\) is within Euclidean distance \(\delta_n\)
of some grid point \(x_j\). The covering number satisfies
\[
N_n\le C\delta_n^{-p}.
\]
Since \(nh^p/\log n\to\infty\), for all sufficiently large \(n\),
\(h^p\ge \log n/n\), and hence \(\log(1/h)=O(\log n)\). Also, because
\(h\to0\), we have \(h^p\le 1\) for all sufficiently large \(n\), so
\[
r_n=\sqrt{\frac{\log n}{nh^p}}
\ge
\sqrt{\frac{\log n}{n}},
\]
which implies \(\log(1/r_n)=O(\log n)\). Therefore
\[
\log N_n
\le C+p\log(1/\delta_n)
=
O(\log n).
\]
For any fixed grid point \(x_j\), set
\[
V_i(x_j):=K_h(X_i-x_j)^2.
\]
As above,
\[
\E V_i(x_j)\le C h^{-p}.
\]
Moreover, since \(K\) is bounded and compactly supported,
\[
|V_i(x_j)|\le C h^{-2p},
\qquad
\E V_i(x_j)^2\le C h^{-3p}.
\]
Bernstein's inequality therefore gives, for a sufficiently large constant \(C_1\),
\[
\Prob\left\{
\frac{1}{n}\sum_{i=1}^n V_i(x_j)>C_1h^{-p}
\right\}
\le
\exp(-cnh^p)
\]
for some constant \(c>0\). Taking a union bound over \(j=1,\dots,N_n\), and using
\[
\log N_n=O(\log n),
\qquad
nh^p/\log n\to\infty,
\]
we obtain
\[
\max_{1\le j\le N_n}
\frac{1}{n}\sum_{i=1}^n K_h(X_i-x_j)^2
=
O_p(h^{-p}).
\]
Equivalently,
\[
\max_{1\le j\le N_n} A_n(x_j)^2
=
\max_{1\le j\le N_n}
\frac{1}{n^2}\sum_{i=1}^n K_h(X_i-x_j)^2
=
O_p\!\left(\frac{1}{nh^p}\right).
\]

Next define
\[
S_n(x_j):=\sup_{z\in\mathcal Z}|G_n(x_j,z)|.
\]
From the fixed-\(x\) argument above, conditional on \(X_1,\dots,X_n\),
\[
\E\{S_n(x_j)\mid X_1,\dots,X_n\}
\le
4A_n(x_j).
\]
Also, changing one observation \(Z_i\) while keeping all other variables fixed changes
\(S_n(x_j)\) by at most \(|a_i(x_j)|\). Therefore, by the bounded-difference inequality, conditional on \(X_1,\dots,X_n\),
\[
\Prob\left\{
S_n(x_j)
>
4A_n(x_j)+t
\,\middle|\, X_1,\dots,X_n
\right\}
\le
\exp\left(
-\frac{t^2}{2A_n(x_j)^2}
\right).
\]
On the event
\[
\max_{1\le j\le N_n} A_n(x_j)^2
\le
\frac{C}{nh^p},
\]
choose \(t=C_2 r_n\) with \(C_2\) large enough. Then
\[
\Prob\left\{
\max_{1\le j\le N_n}S_n(x_j)>C_3r_n
\,\middle|\, X_1,\dots,X_n
\right\}
\le
N_n \exp(-C_4\log n),
\]
which tends to zero for a sufficiently large constant \(C_4\), because
\(\log N_n=O(\log n)\). Hence
\[
\max_{1\le j\le N_n}S_n(x_j)
=
O_p(r_n).
\]

It remains to pass from the grid to all \(x\in\mathcal X\). Let \(x\in\mathcal X\), and choose \(x_j\in\mathcal N_n\) such that
\(\|x-x_j\|\le\delta_n\). Since \(K\) is Lipschitz,
\[
|K_h(X_i-x)-K_h(X_i-x_j)|
\le
L_K h^{-p-1}\|x-x_j\|
\le
L_K h^{-p-1}\delta_n.
\]
Also,
\[
\left|
\one\{Z_i\le z\}-F_Z(z\mid X_i)
\right|
\le 1.
\]
Therefore, uniformly over \(z\in\mathcal Z\),
\begin{align*}
|G_n(x,z)-G_n(x_j,z)|
&\le
\frac{1}{n}\sum_{i=1}^n
|K_h(X_i-x)-K_h(X_i-x_j)| \\
&\le
L_K h^{-p-1}\delta_n
=
L_K r_n.
\end{align*}
Combining this grid approximation with
\[
\max_{1\le j\le N_n}S_n(x_j)=O_p(r_n)
\]
gives
\[
\sup_{x\in\mathcal X,\;z\in\mathcal Z}|G_n(x,z)|
=
O_p(r_n)
=
O_p\!\left(\sqrt{\frac{\log n}{nh^p}}\right).
\]
This proves the uniform claim.
\end{proof}

\subsubsection{Proof of \Cref{prop:kernel_cdf_local_uniform_aux}}

\begin{proof}
Let \(R_K<\infty\) be such that \(K(v)=0\) whenever \(\|v\|>R_K\). Write
\[
\hat f_h(x):=\frac{1}{n}\sum_{i=1}^n K_h(X_i-x),
\qquad
m_h(x):=\E\{\hat f_h(x)\}.
\]
By the design-density and kernel-mass conditions,
\[
m_h(x)
=
\int_{\mathcal X} K_h(t-x)f_X(t)\,dt
\ge c_Xc_{\mathcal X}
\]
for all sufficiently small \(h\). By Lemma~\ref{lem:kernel_denominator_uniform},
\[
\hat f_h(x)-m_h(x)=O_p((nh^p)^{-1/2})=o_p(1).
\]
Hence \(\hat f_h(x)\) is bounded away from zero with probability tending to one.

Decompose
\[
\hat F_Z(z\mid x)-F_Z(z\mid x)
=
S_n(z,x)+B_n(z,x),
\]
where
\[
S_n(z,x)
:=
\sum_{i=1}^n w_i(x;h)
\{\one\{Z_i\le z\}-F_Z(z\mid X_i)\},
\]
and
\[
B_n(z,x)
:=
\sum_{i=1}^n w_i(x;h)
\{F_Z(z\mid X_i)-F_Z(z\mid x)\}.
\]
If \(w_i(x;h)>0\), then \(\|X_i-x\|\le R_Kh\). Therefore, by the H\"older condition on \(x'\mapsto F_Z(z\mid x')\),
\[
\sup_{z\in\mathcal Z}|B_n(z,x)|
\le
C_Z R_K^\beta h^\beta.
\]
Since \(\hat f_h(x)\) is bounded away from zero with probability tending to one, and since
\[
S_n(z,x)=\frac{1}{\hat f_h(x)}G_n(x,z),
\]
Lemma~\ref{lem:kernel_emp_process} gives
\[
\sup_{z\in\mathcal Z}|S_n(z,x)|
\le
\frac{1}{\hat f_h(x)}
\sup_{z\in\mathcal Z}|G_n(x,z)|
=
O_p\bigl((nh^p)^{-1/2}\bigr).
\]
Combining the bounds on \(S_n\) and \(B_n\) yields
\[
\sup_{z\in\mathcal Z}
\bigl|\hat F_Z(z\mid x)-F_Z(z\mid x)\bigr|
=
O_p\bigl(h^\beta+(nh^p)^{-1/2}\bigr).
\]
\end{proof}

\subsubsection{Proof of \Cref{prop:kernel_quantile_rate_aux}}

\begin{proof}
Let
\[
a_n:=h^\beta+(n h^p)^{-1/2}.
\]
By continuity of \(f_{Z\mid X}(z\mid x)\) and the lower bound at \(m_\tau(x)\), there exist \(\delta_0>0\) and \(c_0>0\) such that
\[
f_{Z\mid X}(z\mid x)\ge c_0
\qquad
\text{for all } |z-m_\tau(x)|\le \delta_0.
\]
Fix \(M>0\), and define
\[
m_-(x):=m_\tau(x)-M a_n,
\qquad
m_+(x):=m_\tau(x)+M a_n.
\]
For \(n\) large enough, both points lie in \(\mathcal Z\cap[m_\tau(x)-\delta_0,m_\tau(x)+\delta_0]\). By the mean value theorem,
\[
F_Z(m_+(x)\mid x)-\tau
=
f_{Z\mid X}(\bar m_+(x)\mid x)\,M a_n
\ge
c_0 M a_n,
\]
for some \(\bar m_+(x)\) between \(m_\tau(x)\) and \(m_+(x)\). Similarly,
\[
\tau-F_Z(m_-(x)\mid x)\ge c_0 M a_n.
\]

By Proposition~\ref{prop:kernel_cdf_local_uniform_aux},
\[
\sup_{z\in\mathcal Z}\bigl|\hat F_Z(z\mid x)-F_Z(z\mid x)\bigr|
=
O_p(a_n).
\]
Hence, for \(M\) sufficiently large,
\[
\hat F_Z(m_+(x)\mid x)>\tau,
\qquad
\hat F_Z(m_-(x)\mid x)<\tau,
\]
with probability tending to one. Since \(\hat m_\tau(x)\) is defined by inversion of \(\hat F_Z(\cdot\mid x)\), it follows that
\[
m_-(x)\le \hat m_\tau(x)\le m_+(x)
\]
with probability tending to one. Therefore
\[
\hat m_\tau(x)-m_\tau(x)
=
O_p(a_n)
=
O_p\!\bigl(h^\beta+(n h^p)^{-1/2}\bigr).
\]

For the second claim, the leading mean squared error is
\[
h^{2\beta}+(n h^p)^{-1}.
\]
Balancing squared bias and variance gives
\[
h^\star \asymp n^{-1/(2\beta+p)}.
\]
Substituting this into the previous rate yields
\[
\hat m_\tau(x)-m_\tau(x)
=
O_p\!\left(n^{-\beta/(2\beta+p)}\right).
\]
\end{proof}

\subsubsection{Proof of \Cref{prop:kernel_cdf_uniform_aux}}

\begin{proof}
Let \(R_K<\infty\) be such that \(K(v)=0\) whenever \(\|v\|>R_K\). Write
\[
\hat f_h(x):=\frac{1}{n}\sum_{i=1}^n K_h(X_i-x),
\qquad
m_h(x):=\E\{\hat f_h(x)\}.
\]

By Lemma~\ref{lem:kernel_denominator_uniform},
\[
\sup_{x\in\mathcal X}|\hat f_h(x)-m_h(x)|
=
O_p\!\left(\sqrt{\frac{\log n}{nh^p}}\right).
\]
The design-density and kernel-mass conditions imply
\[
\inf_{x\in\mathcal X}m_h(x)\ge c_Xc_{\mathcal X}
\]
for all sufficiently small \(h\). Since \(nh^p/\log n\to\infty\), it follows that
\[
\inf_{x\in\mathcal X}\hat f_h(x)
\ge c_Xc_{\mathcal X}/2
\]
with probability tending to one.

For every \(x\in\mathcal X\), decompose
\[
\hat F_Z(z\mid x)-F_Z(z\mid x)
=
S_n(z,x)+B_n(z,x),
\]
where
\[
S_n(z,x)
:=
\sum_{i=1}^n w_i(x;h)
\{\one\{Z_i\le z\}-F_Z(z\mid X_i)\},
\]
and
\[
B_n(z,x)
:=
\sum_{i=1}^n w_i(x;h)
\{F_Z(z\mid X_i)-F_Z(z\mid x)\}.
\]
If \(w_i(x;h)>0\), then \(\|X_i-x\|\le R_Kh\). Therefore,
\[
\sup_{x\in\mathcal X,\;z\in\mathcal Z}|B_n(z,x)|
\le
C_ZR_K^\beta h^\beta.
\]
For the stochastic term,
\[
S_n(z,x)=\frac{1}{\hat f_h(x)}G_n(x,z).
\]
Since \(\inf_{x\in\mathcal X}\hat f_h(x)\) is bounded away from zero with
probability tending to one, Lemma~\ref{lem:kernel_emp_process} gives
\[
\sup_{x\in\mathcal X,\;z\in\mathcal Z}|S_n(z,x)|
\le
\left(\inf_{x\in\mathcal X}\hat f_h(x)\right)^{-1}
\sup_{x\in\mathcal X,\;z\in\mathcal Z}|G_n(x,z)|
=
O_p\!\left(\sqrt{\frac{\log n}{nh^p}}\right).
\]
Combining the bounds on \(S_n\) and \(B_n\) gives
\[
\sup_{x\in\mathcal X,\; z\in \mathcal Z}
\bigl|\hat F_Z(z\mid x)-F_Z(z\mid x)\bigr|
=
O_p\!\left(h^\beta+\sqrt{\frac{\log n}{nh^p}}\right).
\]
\end{proof}

\subsection{Wrapped response distribution}
\label{app:generated_pseudo}

The estimator \(\hat u_\tau(x)\) is constructed from the generated pseudo-responses
\[
\hat U_i:=\hat F^L(Y_i^H\mid X_i^H),
\]
rather than the oracle responses
\[
U_i:=F^L(Y_i^H\mid X_i^H).
\]
The following lemma controls the effect of this perturbation.

Define the oracle weighted CDF
\[
\bar F^U(u\mid x):=\sum_{i=1}^{n_H} w_i^U(x)\,\one\{U_i\le u\},
\]
and the corresponding oracle quantile estimator
\[
\bar u_\tau(x):=\inf\{u:\bar F^U(u\mid x)\ge \tau\}.
\]
Also define
\[
\hat F^U(u\mid x):=\sum_{i=1}^{n_H} w_i^U(x)\,\one\{\hat U_i\le u\},
\qquad
\hat u_\tau(x):=\inf\{u:\hat F^U(u\mid x)\ge \tau\}.
\]

\begin{lemma}
\label{lem:u_generated}
If
\[
\sup_{x\in\mathcal X,\; y\in \mathcal Y}
\bigl|\hat F^L(y\mid x)-F^L(y\mid x)\bigr|
=
O_p\!\left(
 h_L^{\beta_L}
 +
 \sqrt{\frac{\log n_L}{n_L h_L^p}}
\right),
\]
then
\[
\hat u_\tau(x)-\bar u_\tau(x)
=
O_p\!\left(
 h_L^{\beta_L}
 +
 \sqrt{\frac{\log n_L}{n_L h_L^p}}
\right).
\]
Consequently, if
\[
\bar u_\tau(x)-u_\tau(x)=O_p(s_n),
\]
then
\[
\hat u_\tau(x)-u_\tau(x)
=
O_p\!\left(
 s_n
 +
 h_L^{\beta_L}
 +
 \sqrt{\frac{\log n_L}{n_L h_L^p}}
\right).
\]
\end{lemma}

\begin{proof}
Let
\[
\Delta_n:=\max_{1\le i\le n_H}|\hat U_i-U_i|.
\]
Since \(Y_i^H\in \mathcal Y\) almost surely under Assumption~\ref{ass:global_locality},
\[
\Delta_n
=
\max_{1\le i\le n_H}
\bigl|\hat F^L(Y_i^H\mid X_i^H)-F^L(Y_i^H\mid X_i^H)\bigr|
\le
\sup_{x\in\mathcal X,\; y\in \mathcal Y}
\bigl|\hat F^L(y\mid x)-F^L(y\mid x)\bigr|.
\]
By Proposition~\ref{prop:kernel_cdf_uniform_aux},
\[
\Delta_n
=
O_p\!\left(
h_L^{\beta_L}
+
\sqrt{\frac{\log n_L}{n_L h_L^p}}
\right).
\]

For every \(u\in\mathbb R\),
\[
\one\{U_i\le u-\Delta_n\}\le \one\{\hat U_i\le u\}\le \one\{U_i\le u+\Delta_n\}.
\]
Multiplying by the nonnegative weights \(w_i^U(x)\) and summing over \(i\) gives
\[
\bar F^U(u-\Delta_n\mid x)\le \hat F^U(u\mid x)\le \bar F^U(u+\Delta_n\mid x).
\]
By inversion of these monotone functions,
\[
\bar u_\tau(x)-\Delta_n\le \hat u_\tau(x)\le \bar u_\tau(x)+\Delta_n.
\]
Hence
\[
|\hat u_\tau(x)-\bar u_\tau(x)|\le \Delta_n
=
O_p\!\left(
h_L^{\beta_L}
+
\sqrt{\frac{\log n_L}{n_L h_L^p}}
\right).
\]
The second claim follows from the triangle inequality.
\end{proof}

\subsection{Proof of Corollary~\ref{cor:wrapper_rate_transfer}}

\begin{proof}
By Assumption~\ref{ass:kernel_rate_setup}, we have \(\beta_q,\beta_u,\beta_L\in(0,1]\). The bounded-density and kernel-mass conditions in Assumption~\ref{ass:kernel_rate_setup} are exactly the covariate-design conditions required by Propositions~\ref{prop:kernel_quantile_rate_aux} and \ref{prop:kernel_cdf_uniform_aux} for each application below.

Applying Proposition~\ref{prop:kernel_quantile_rate_aux} with \(Z=Y^H\), \(m_\tau=q_\tau\), and \(\beta=\beta_q\) gives
\[
\hat q_\tau^H(x)-q_\tau(x)
=
O_p\!\left(n_H^{-\beta_q/(2\beta_q+p)}\right),
\]
when \(h_q\) is chosen at the MSE-optimal order \(h_q^\star\asymp n_H^{-1/(2\beta_q+p)}\).

Next define the oracle pseudo-responses
$U_i:=F^L(Y_i^H\mid X_i^H),$
and let \(\bar u_\tau(x)\) be the local constant kernel quantile estimator constructed from \(\{(X_i^H,U_i)\}_{i=1}^{n_H}\). By Assumption~\ref{ass:kernel_rate_setup}(ii), the wrapped response \(U\) satisfies the conditions of Proposition~\ref{prop:kernel_quantile_rate_aux} with \(\mathcal Z=\mathcal U^L\), \(m_\tau=u_\tau\), and \(\beta=\beta_u\). Therefore,
\[
\bar u_\tau(x)-u_\tau(x)
=
O_p\!\left(n_H^{-\beta_u/(2\beta_u+p)}\right),
\]
when \(h_u\) is chosen at the MSE-optimal order \(h_u^\star\asymp n_H^{-1/(2\beta_u+p)}\).

Applying Proposition~\ref{prop:kernel_cdf_uniform_aux} to the LF sample with bandwidth
\[
h_L^\star \asymp \left(\frac{\log n_L}{n_L}\right)^{1/(2\beta_L+p)},
\]
gives
\[
\sup_{x\in\mathcal X,\; y\in \mathcal Y}
\bigl|\hat F^L(y\mid x)-F^L(y\mid x)\bigr|
=
O_p\!\left(
\left(\frac{\log n_L}{n_L}\right)^{\beta_L/(2\beta_L+p)}
\right).
\]
Hence, by Lemma~\ref{lem:u_generated},
\[
\hat u_\tau(x)-\bar u_\tau(x)
=
O_p\!\left(
\left(\frac{\log n_L}{n_L}\right)^{\beta_L/(2\beta_L+p)}
\right),
\]
and therefore
\[
\hat u_\tau(x)-u_\tau(x)
=
O_p\!\left(
 n_H^{-\beta_u/(2\beta_u+p)}
 +
 \left(\frac{\log n_L}{n_L}\right)^{\beta_L/(2\beta_L+p)}
\right).
\]

Therefore, setting
\[
a_n:=n_H^{-\beta_u/(2\beta_u+p)}
+
\left(\frac{\log n_L}{n_L}\right)^{\beta_L/(2\beta_L+p)},
\qquad
b_n:=\left(\frac{\log n_L}{n_L}\right)^{\beta_L/(2\beta_L+p)},
\]
and applying Theorem~\ref{thm:wrapper_nonparam} yields
\[
\tilde q_\tau(x)-q_\tau(x)
=
O_p\!\left(
 n_H^{-\beta_u/(2\beta_u+p)}
 +
 \left(\frac{\log n_L}{n_L}\right)^{\beta_L/(2\beta_L+p)}
\right).
\]
\end{proof}

\subsection{Proof of Theorem~\ref{thm:multiple_bias_nonparam}}

\begin{proof}
Fix \(x\) and \(\tau\), and suppress them from the notation. Write
\[
q:=q_\tau(x),\qquad
\tilde q:=\tilde q_\tau(x),\qquad
e:=e_\tau(x)=\tilde q-q,
\]
and
\[
F:=F^H(\cdot\mid x),\qquad
f:=f^H(\cdot\mid x),\qquad
\nu:=\nu_\tau(x),\qquad
\eta:=\eta_\tau(x).
\]
By definition,
\[
\tilde q_{\tau,\gamma}(x)-q
=
e-\gamma\frac{\hat F^H(\tilde q\mid x)-\tau}{\hat f^H(\tilde q\mid x)}.
\]

Since \(F(q)=\tau\) and \(F\) is twice continuously differentiable,
\[
F(\tilde q)-\tau
=
f(q)e+O_p(e^2).
\]
Hence
\[
\hat F^H(\tilde q\mid x)-\tau
=
F(\tilde q)-\tau+\nu
=
f(q)e+\nu+O_p(e^2).
\]

Likewise, since \(f\) is continuously differentiable,
\[
f(\tilde q)=f(q)+O_p(|e|),
\quad 
\hat f^H(\tilde q\mid x)=f(\tilde q)+\eta=f(q)+\eta+O_p(|e|).
\]
Because \(e=o_p(1)\) and \(\eta=o_p(1)\) by Assumption~\ref{ass:pilot_local}, we have
\[
\hat f^H(\tilde q\mid x)=f(q)+o_p(1),
\]
so \(\hat f^H(\tilde q\mid x)\) is bounded away from zero in probability. The reciprocal expansion then yields
\[
\frac{1}{\hat f^H(\tilde q\mid x)}
=
\frac{1}{f(q)}+O_p(|e|+|\eta|).
\]

Multiplying the numerator and denominator expansions,
\begin{align*}
\frac{\hat F^H(\tilde q\mid x)-\tau}{\hat f^H(\tilde q\mid x)}
&=
\bigl(f(q)e+\nu+O_p(e^2)\bigr)
\Bigl(\frac{1}{f(q)}+O_p(|e|+|\eta|)\Bigr) \\
&=
e+\frac{\nu}{f(q)}
+
O_p\!\bigl(e^2+e\eta+e\nu+\nu\eta\bigr).
\end{align*}
Substituting this into the definition of \(\tilde q_{\tau,\gamma}(x)-q\) gives
\[
\tilde q_{\tau,\gamma}(x)-q
=
(1-\gamma)e
-\gamma\frac{\nu}{f(q)}
+
O_p\!\bigl(e^2+e\eta+e\nu+\nu\eta\bigr),
\]
which is exactly the stated expansion.
\end{proof}

\subsection{Proof of Proposition~\ref{prop:iterated_onestep_main}}
\label{sect:proof_steps}

\begin{proof}
Work on the event in Assumption~\ref{ass:multistep}, whose probability tends to
one. 
Let \(C_N:=L/2c\). By the assumption, \(e_0\le \delta_0\) and
\(C_N\delta_0<1\).
Write
\[
\Psi(q):=\Psi_n(q;x),\qquad
\hat q:=\hat q_\tau^H(x),\qquad
q_m:=q_\tau^{(m)}(x),
\qquad
e_m:=|q_m-\hat q|.
\]

We first show by induction that \(e_m\le \delta_0\) for all \(m\ge 0\). The
claim holds for \(m=0\). Suppose it holds for some \(m\). Then
\(q_m\in[\hat q-\delta_0,\hat q+\delta_0]\subset I_x\), so Taylor's theorem can
be applied on \(I_x\). Since \(\Psi(\hat q)=0\), there exists a point
\(\xi_m\) between \(q_m\) and \(\hat q\) such that
\[
0
=
\Psi(q_m)+\Psi'(q_m)(\hat q-q_m)
+\frac{1}{2}\Psi''(\xi_m)(\hat q-q_m)^2.
\]
Because
\[
\Psi'(q)=\hat f^H(q\mid x),
\qquad
q_{m+1}=q_m-\frac{\Psi(q_m)}{\Psi'(q_m)},
\]
we obtain
\[
q_{m+1}-\hat q
=
\frac{\Psi''(\xi_m)}{2\Psi'(q_m)}(q_m-\hat q)^2.
\]
By Assumption~\ref{ass:multistep},
\[
|\Psi''(\xi_m)|
=
|\partial_q\hat f^H(\xi_m\mid x)|
\le L,
\qquad
\Psi'(q_m)=\hat f^H(q_m\mid x)\ge c.
\]
Therefore $e_{m+1}\le C_N e_m^2$.
Using the induction hypothesis \(e_m\le\delta_0\), we get
\[
e_{m+1}
\le C_N e_m^2
\le C_N\delta_0 e_m
< e_m
\le \delta_0.
\]
Thus \(e_m\le\delta_0\) for all \(m\ge0\), and all iterates remain in \(I_x\).

Moreover, again by induction,
\[
e_m\le C_N^{-1}(C_N\delta_0)^{2^m}.
\]
Indeed, for \(m=0\),
\[
e_0\le \delta_0=C_N^{-1}(C_N\delta_0).
\]
If the bound holds at step \(m\), then
\[
e_{m+1}
\le C_N e_m^2
\le
C_N\Bigl\{C_N^{-1}(C_N\delta_0)^{2^m}\Bigr\}^2
=
C_N^{-1}(C_N\delta_0)^{2^{m+1}}.
\]
Since \(C_N\delta_0<1\), it follows that \(e_m\to0\). Hence
\[
q_\tau^{(m)}(x)\to \hat q_\tau^H(x)
\qquad\text{as }m\to\infty.
\]
Because the event in Assumption~\ref{ass:multistep} has probability tending to
one, the stated conclusion follows.
\end{proof}

\end{document}